\documentclass[%
 reprint,
 superscriptaddress,
nofootinbib,
 amsmath,amssymb,
 aps,
]{revtex4-2}

\usepackage{graphicx}
\usepackage{dcolumn}
\usepackage{bm}
\usepackage[bookmarks, colorlinks=true, urlcolor=blue, linkcolor=black, citecolor=blue]{hyperref}
\usepackage{mathtools}
\usepackage{physics}
\usepackage{enumitem}
\usepackage[dvipsnames]{xcolor}
\usepackage{float}
\usepackage{amsthm}
\newtheorem{thm}{Theorem}
\newtheorem{theorem}{Theorem}[section]
\newtheorem{lemma}[theorem]{Lemma}
\newtheorem{proposition}{Proposition}
\newtheorem{corollary}{Corollary}

\begin{document}

\preprint{APS/123-QED}

\title{Quantum Machine Learning Beyond Kernel Methods}

\author{Sofiene Jerbi}
\affiliation{Institute for Theoretical Physics, University of Innsbruck, Technikerstr. 21a, A-6020 Innsbruck, Austria}
\author{Lukas J. Fiderer}
\affiliation{Institute for Theoretical Physics, University of Innsbruck, Technikerstr. 21a, A-6020 Innsbruck, Austria}
\author{Hendrik Poulsen Nautrup}
\affiliation{Institute for Theoretical Physics, University of Innsbruck, Technikerstr. 21a, A-6020 Innsbruck, Austria}
\author{Jonas M. K\"{u}bler}
\affiliation{Max Planck Institute for Intelligent Systems, T\"{u}bingen, Germany}
\author{\\Hans J. Briegel}
\affiliation{Institute for Theoretical Physics, University of Innsbruck, Technikerstr. 21a, A-6020 Innsbruck, Austria}
\author{Vedran Dunjko}
\affiliation{Leiden University, Niels Bohrweg 1, 2333 CA Leiden, Netherlands}

\date{\today}

\begin{abstract}
Machine learning algorithms based on parametrized quantum circuits are prime candidates for near-term applications on noisy quantum computers. In this direction, various types of quantum machine learning models have been introduced and studied extensively. Yet, our understanding of how these models compare, both mutually and to classical models, remains limited. In this work, we identify a constructive framework that captures all standard models based on parametrized quantum circuits: that of linear quantum models. In particular, we show using tools from quantum information theory how data re-uploading circuits, an apparent outlier of this framework, can be efficiently mapped into the simpler picture of linear models in quantum Hilbert spaces. Furthermore, we analyze the experimentally-relevant resource requirements of these models in terms of qubit number and amount of data needed to learn. Based on recent results from classical machine learning, we prove that linear quantum models must utilize exponentially more qubits than data re-uploading models in order to solve certain learning tasks, while kernel methods additionally require exponentially more data points. Our results provide a more comprehensive view of quantum machine learning models as well as insights on the compatibility of different models with NISQ constraints.
\end{abstract}

\maketitle

\section{Introduction}

In the current Noisy Intermediate-Scale Quantum (NISQ) era \cite{preskill18}, a few methods have been proposed to construct useful quantum algorithms that are compatible with mild hardware restrictions \cite{cerezo20,bharti21}. Most of these methods involve the specification of a quantum circuit Ansatz, optimized in a classical fashion to solve specific computational tasks. Next to variational quantum eigensolvers in chemistry \cite{peruzzo14} and variants of the quantum approximate optimization algorithm \cite{farhi14}, machine learning approaches based on such parametrized quantum circuits \cite{benedetti19} stand as some of the most promising practical applications to yield quantum advantages.

In essence, a supervised machine learning problem often reduces to the task of fitting a parametrized function -- also referred to as the machine learning model -- to a set of previously labeled points, called a training set. Interestingly, many problems in physics and beyond, from the classification of phases of matter \cite{carrasquilla17} to predicting the folding structures of proteins \cite{jumper21}, can be phrased as such machine learning tasks. In the domain of quantum machine learning \cite{biamonte17,dunjko18}, an emerging approach for this type of problem is to use parametrized quantum circuits to define a hypothesis class of functions \cite{schuld20,farhi18,liu18,zhu19,skolik21,jerbi21}. The hope is for these parametrized models to offer representational power beyond what is possible with classical models, including the highly successful deep neural networks. And indeed, we have substantial evidence of such a quantum learning advantage for artificial problems \cite{liu20,du20,sweke20,huang20,jerbi21,huang21}, but the next frontier is to show that quantum models can be advantageous in solving real-world problems as well.\break Yet, it is still unclear which of these models we should preferably use in practical applications. To bring quantum machine learning models forward, we first need a deeper understanding of their learning performance guarantees and the actual resource requirements they entail.

\begin{figure}[!t]
\begin{center}
\includegraphics[width=\linewidth]{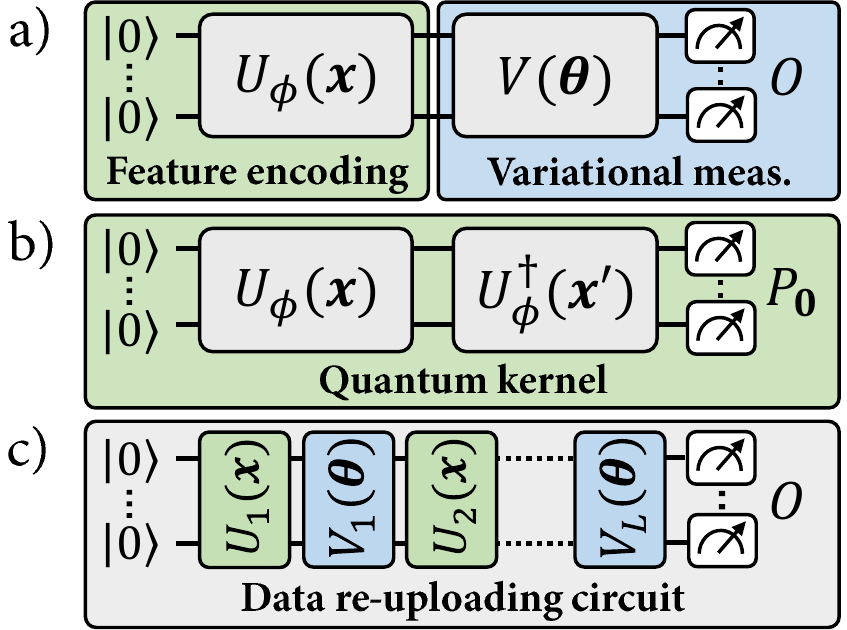}
\caption{The quantum machine learning models studied in this work. a) An explicit quantum model, where the label of a data point $\bm{x}$ is specified by the expectation value of a variational measurement on its associated quantum feature state $\rho(\bm{x})$. b) The quantum kernel associated to these quantum feature states. The expectation value of the projection $P_{\bm{0}}=\ket{\bm{0}}\!\!\bra{\bm{0}}$\break corresponds to the inner product between $\rho(\bm{x})$ and $\rho(\bm{x'})$. An implicit quantum model is defined by a linear combination of such inner products, for $\bm{x}$ an input point and $\bm{x'}$ training data points. c) A data re-uploading model, interlaying data encoding and variational unitaries before a final measurement.}
\label{fig:quantum-models}
\end{center}
\vspace{-2em}
\end{figure}

\begin{figure*}
\begin{center}
\includegraphics[width=0.77\linewidth]{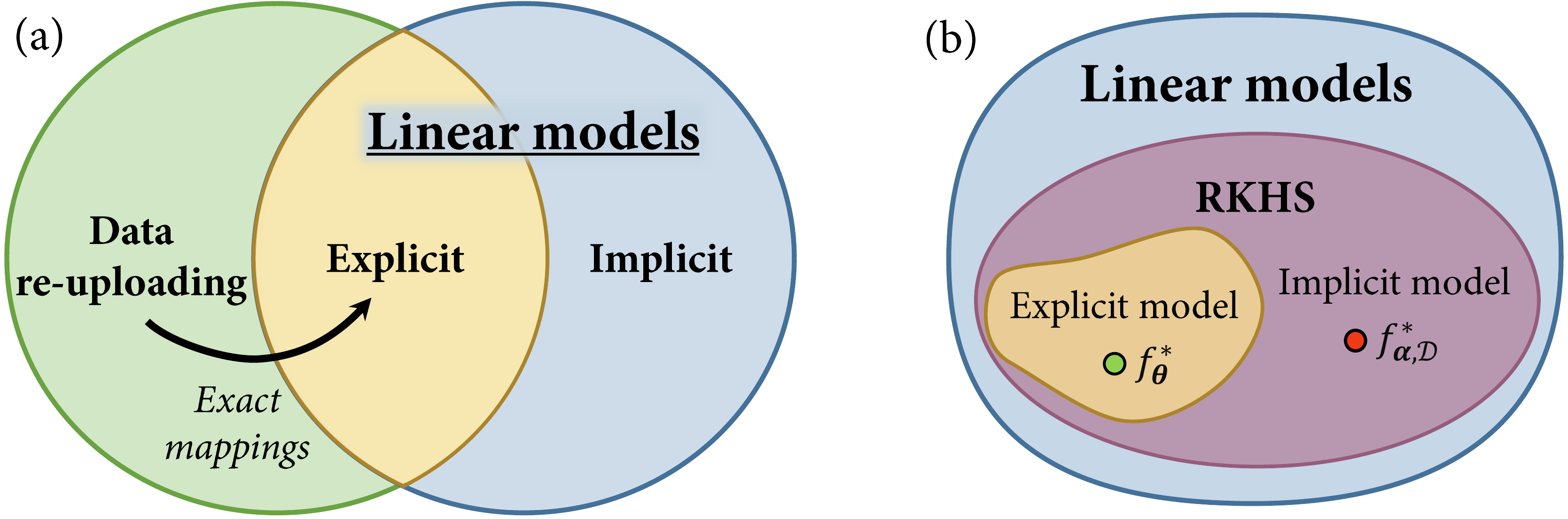}
\caption{The model families in quantum machine learning. (a) While data re-uploading models are by definition a generalization of linear quantum models, our exact mappings (see Sec.~\ref{sec:linear-reup}) demonstrate that any polynomial-size data re-uploading model can be realized by a polynomial-size explicit linear model. (b) Kernelizing an explicit model corresponds to turning its observable into a linear combination of feature states $\rho(\bm{x})$, for $\bm{x}$ in a dataset $\mathcal{D}$. The representer theorem (see Sec.~\ref{sec:rkhs-representer}) guarantees that, for any dataset $\mathcal{D}$, the implicit model $f^*_{\bm{\alpha},\mathcal{D}}$ minimizing the training loss associated to $\mathcal{D}$ outperforms any explicit minimizer $f^*_{\bm{\theta}}$ from the same Reproducing Kernel Hilbert Space (RKHS) \emph{with respect to this same training loss}. However, depending on the feature encoding $\rho(\cdot)$ and the data distribution, a restricted dataset $\mathcal{D}$ may cause the implicit minimizer $f^*_{\bm{\alpha},\mathcal{D}}$ to severely overfit on the dataset and have dramatically worse generalization performance than $f^*_{\bm{\theta}}$ (see Sec.~\ref{sec:explicit-vs-implicit} and \ref{sec:num-comp}).}
\label{fig:summary-results_1}
\end{center}
\vspace{-2em}
\end{figure*}

This is precisely where the main contribution of our work lies. In this paper, we analyze the relations between the different quantum models currently proposed in the literature, and uncover clear indications on which models to use in practice, in light of experimentally-relevant restrictions such as the number of qubits and quantum circuit evaluations needed to learn.

Previous works have made strides in this direction by exploiting a connection between some quantum models and kernel methods from classical machine learning \cite{scholkopf02}. Many quantum models indeed operate by encoding data in a high-dimensional Hilbert space and using solely inner products evaluated in this feature space to model properties of the data. This is also how kernel methods work. Building on this similarity, the authors of \cite{havlivcek19,schuld19} noted that a given quantum encoding can be used to define two types of models (see Fig.~\ref{fig:quantum-models}): (a) \emph{explicit} quantum models, where an encoded data point is measured according to a variational observable that specifies its label, or\break (b) \emph{implicit} kernel models, where weighted inner products of encoded data points are used to assign labels instead. In the quantum machine learning literature, much emphasis has been placed on implicit models \cite{schuld21,lloyd20,huang20,kubler21,peters21,haug21,bartkiewicz20,kusumoto21}, in part due to a fundamental result known as the representer theorem \cite{scholkopf02}. This result shows that implicit models can always achieve a smaller labeling error than explicit models, when evaluated on the same training set. Seemingly, this suggests that implicit models are systematically more advantageous than their explicit counterparts in solving machine learning tasks \cite{schuld21}. This idea also inspired a line of research where, in order to evaluate the existence of quantum advantages, classical models were only compared to quantum kernel methods. This restricted comparison led to the conclusion that classical models could be competitive with (or outperform) quantum models, even in tailored quantum problems \cite{huang20}.

In recent times, there has also been progress in so-called \emph{data re-uploading} models \cite{perez20} which have demonstrated their importance in designing expressive models, both analytically \cite{schuld21b} and empirically \cite{perez20,jerbi21,skolik21}, and proving that (even single-qubit) parametrized quantum circuits are universal function approximators \cite{perez21,goto20}. Through their alternation of data-encoding and variational unitaries, data re-uploading models can be seen as a generalization of explicit models. However, this generalization also breaks the correspondence to implicit models, as a given data point $\bm{x}$ no longer corresponds to a \emph{fixed} encoded point $\rho(\bm{x})$. Hence, these observations suggest that data re-uploading models are strictly more general than explicit models and that they are incompatible with the kernel-model paradigm. Until now, it remained an open question whether some advantage could be gained from data re-uploading models, in light of the guarantees of kernel methods.

In this work, we introduce a unifying framework for explicit, implicit and data re-uploading quantum models (see Fig.~\ref{fig:summary-results_1}). We show that all function families stemming from these can be formulated as linear models in suitably-defined quantum feature spaces. This allows us to systematically compare explicit and data re-uploading models to their kernel formulations. We find that, while kernel models are guaranteed to achieve a lower training error, this improvement can come at the cost of a poor generalization performance outside the training set. Our results indicate that the advantages of quantum machine learning may lie beyond kernel methods, more specifically in explicit and data re-uploading models.
To corroborate this theory, we quantify the resource requirements of these different quantum models in terms of the number of qubits and data points needed to learn. We show the existence of a regression task with exponential separations between each pair of quantum models, demonstrating the practical advantages of explicit models over implicit models, and of data re-uploading models over explicit models. From an experimental perspective, these separations shed light on the resource efficiency of different quantum models, which is of crucial importance for near-term applications in quantum machine learning.

\section{A unifying framework for quantum learning models}

In this section, we start by reviewing the notion of linear quantum models and explain how explicit and implicit models are by definition linear models in quantum feature spaces. We then present data re-uploading models and show how, despite being defined as a generalization of explicit models, they can also be realized by linear models in larger Hilbert spaces.\vspace{-0.5em}

\subsection{Linear quantum models}

Let us first understand how explicit and implicit quantum models can both be described as linear quantum models \cite{schuld21,gyurik21}. To define both of these models, we first consider a feature encoding unitary $U_{\phi}:\mathcal{X}\rightarrow\mathcal{F}$ that maps input vectors $\bm{x} \in \mathcal{X}$, e.g., images in $\mathbb{R}^d$, to $n$-qubit quantum states $\rho(\bm{x}) = U_{\phi}(\bm{x}) \ket{\bm{0}}\!\!\bra{\bm{0}} U_{\phi}^\dagger(\bm{x})$ in the Hilbert space $\mathcal{F}$ of $2^n\times2^n$ Hermitian operators.

A linear function in the quantum feature space $\mathcal{F}$ is defined by the expectation values
\begin{equation}\label{eq:linear-model}
    f(\bm{x}) = \text{Tr}[\rho(\bm{x})O],
\end{equation}
for some Hermitian observable $O\in\mathcal{F}$. Indeed, one can see from Eq.~(\ref{eq:linear-model}) that $f(\bm{x})$ is the Hilbert-Schmidt inner product between the Hermitian matrices $\rho(\bm{x})$ and $O$,\break which is by definition a linear function of the form $\langle\phi(\bm{x}), w\rangle_{\mathcal{F}}$, for $\phi(\bm{x}) = \rho(\bm{x})$ and $w = O$. In a regression task, these real-valued expectation values are used directly to define a labeling function, while in a classification task, they are post-processed to produce discrete labels (using, for instance, a sign function).

Explicit and implicit models differ in the way they define the family of observables $\{O\}$ they each consider.\vspace{-1.5em}

\subsubsection{Explicit models}

An explicit quantum model \cite{havlivcek19,schuld19} using the feature encoding $U_{\phi}(\bm{x})$ is defined by a variational family of unitaries $V(\bm{\theta})$ and a fixed observable $O$, such that
\begin{equation}\label{eq:explicit-model}
    f_{\bm{\theta}}(\bm{x}) = \text{Tr}[\rho(\bm{x})O_{\bm{\theta}}],
\end{equation}
for $O_{\bm{\theta}} = V(\bm{\theta})^\dagger O V(\bm{\theta})$, specify its labeling function. Restricting the family of variational observables $\{O_{\bm{\theta}}\}_{\bm{\theta}}$ is equivalent to restricting the vectors $w$ accessible to the linear quantum model $f(\bm{x})\!=\!\langle\phi(\bm{x}), w\rangle_{\mathcal{F}},\ w\!\in\!\mathcal{F}$, associated to the encoding $\rho(\bm{x})$.\vspace{-1.5em}

\subsubsection{Implicit models}

Implicit quantum models \cite{havlivcek19,schuld19} are constructed from the quantum feature states $\rho(\bm{x})$ in a different way. Their definition depends directly on the data points $\{\bm{x}^{(1)}, \ldots, \bm{x}^{(M)}\}$ in a given training set $\mathcal{D}$, as they take the form of a linear combination
\begin{equation}\label{eq:kernel-model}
    f_{\bm{\alpha},\mathcal{D}}(\bm{x}) = \sum_{m=1}^{M} \alpha_m k(\bm{x},\bm{x}^{(m)}),
\end{equation}
for $k(\bm{x},\bm{x}^{(m)}) = \langle \phi(\bm{x}),\phi(\bm{x}^{(m)})\rangle_{\mathcal{F}} = \text{Tr}[\rho(\bm{x})\rho(\bm{x}^{(m)})]$ the \emph{kernel function} associated to the feature encoding $U_\phi(\bm{x})$. By linearity of the trace however, we can express any such implicit model as a linear model in $\mathcal{F}$, defined by the observable:
\begin{equation}\label{eq:implicit-obs}
    O_{\bm{\alpha},\mathcal{D}} = \sum_{m=1}^{M} \alpha_m \rho(\bm{x}^{(m)}).
\end{equation}
Therefore, both explicit and implicit quantum models belong to the general family of linear models in the quantum feature space $\mathcal{F}$.

\subsection{Linear realizations of data re-uploading models\label{sec:linear-reup}}

Data re-uploading models \cite{perez20} on the other hand do not naturally fit this formulation. These models generalize explicit models by increasing the number of encoding layers $U_\ell(\bm{x}), 1\leq \ell \leq L$ (which can be all distinct), and interlaying them with variational unitaries $V_\ell(\bm{\theta})$. This results in expectation-value functions of the form:
\begin{equation}\label{eq:dr-model}
 f_{\bm{\theta}}(\bm{x}) = \text{Tr}[\rho_{\bm{\theta}}(\bm{x})O_{\bm{\theta}}],
\end{equation}
for a \emph{variational} encoding $\rho_{\bm{\theta}}(\bm{x})\! =\! U(\bm{x},\bm{\theta}) \ket{\bm{0}}\!\!\bra{\bm{0}} U^\dagger(\bm{x},\bm{\theta})$,\break where $U(\bm{x},\bm{\theta}) = U_L(\bm{x})\prod_{\ell=1}^{L-1} V_\ell(\bm{\theta}) U_\ell(\bm{x})$, and a variational observable $O_{\bm{\theta}} = V_L(\bm{\theta})^\dagger O V_L(\bm{\theta})$.
Given that the unitaries $U_\ell(\bm{x})$ and $V_{\ell'}(\bm{\theta})$ do not commute in general, one cannot straightforwardly gather \emph{all} trainable gates in a final variational observable $O'_{\bm{\theta}} \in \mathcal{F}$ as to obtain a linear model $\tilde{f}_{\bm{\theta}}(\bm{x})=\langle\phi(\bm{x}), O'_{\bm{\theta}}\rangle_{\mathcal{F}}$ with a \emph{fixed} quantum feature encoding $\phi(\bm{x})$.
Our first contribution is to show that, by augmenting the dimension of the Hilbert space $\mathcal{F}$ (i.e., considering circuits that act on a larger number of qubits), one can construct such explicit linear realizations $\tilde{f}_{\bm{\theta}}$ of data re-uploading models. That is, given a family of data re-uploading models $\{f_{\bm{\theta}}(\cdot) = \text{Tr}[\rho_{\bm{\theta}}(\cdot)O_{\bm{\theta}}]\}_{\bm{\theta}}$, we can construct an equivalent family of explicit models $\{\tilde{f}_{\bm{\theta}}(\cdot) = \text{Tr}[\rho'(\cdot)O'_{\bm{\theta}}]\}_{\bm{\theta}}$ that represents all functions in the original family, along with an efficient procedure to map the former models to the latter.

\subsubsection{Approximate mapping}

\begin{figure}[t]
\begin{center}
\includegraphics[width=0.92\linewidth]{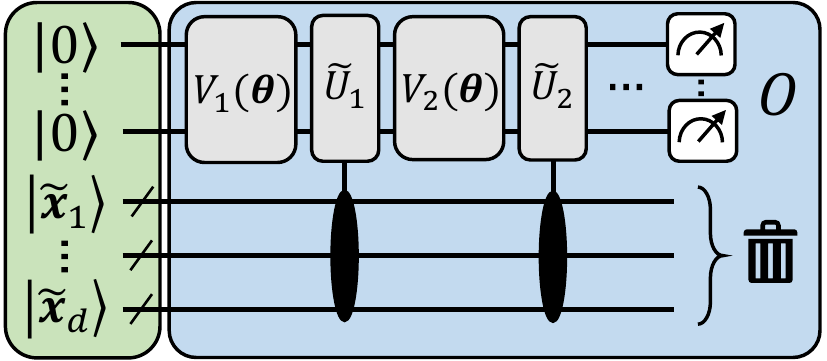}
\caption{An illustrative explicit model approximating a data re-uploading circuit. The circuit acts $n$ working qubits and $dp$ encoding qubits. Pauli-X rotations encode bit-string descriptions $\widetilde{\bm{x}}_i \in \{0,1\}^p$ of the $d$ input components $x_i \in \mathbb{R}$, which constitutes the feature encoding of the explicit model. Fixed and data-independent controlled rotations, interlaid with arbitrary variational unitaries, and a final measurement of the working qubits can result in a good approximation of any parametrized quantum circuit acting on $n$ qubits.}
\label{fig:construction}
\end{center}
\vspace{-2em}
\end{figure}

Before getting to the main result of this section (Theorem \ref{thm:exact-mapping}), we first present an illustrative construction to convey intuition on how mappings from data re-uploading to explicit models can be realized. This construction, depicted in Fig.~\ref{fig:construction}, leads to \emph{approximate} mappings, meaning that these only guarantee $|\tilde{f}_{\bm{\theta}}(\bm{x}) - f_{\bm{\theta}}(\bm{x})| \leq \delta,$ $\forall \bm{x}, \bm{\theta}$ for some (adjustable) error of approximation $\delta$. More precisely, we have:

\begin{proposition}\label{prop:approx-mapping}
Given an arbitrary data re-uploading model $f_{\bm{\theta}}(\bm{x}) = \textnormal{Tr}[\rho_{\bm{\theta}}(\bm{x})O_{\bm{\theta}}]$ as specified by Eq.~(\ref{eq:dr-model}), and an approximation error $\delta>0$, there exists a mapping that produces an explicit model $\tilde{f}_{\bm{\theta}}(\bm{x}) = \textnormal{Tr}[\rho'(\bm{x})O'_{\bm{\theta}}]$ as specified by Eq.~(\ref{eq:explicit-model}), such that:
\begin{equation}
|\textnormal{Tr}[\rho'(\bm{x})O'_{\bm{\theta}}] - \textnormal{Tr}[\rho_{\bm{\theta}}(\bm{x})O_{\bm{\theta}}]|\leq \delta,\ \forall \bm{x}, \bm{\theta}.
\end{equation}
For $D$ the number of encoding gates used by the data re-uploading model and $\norm{O}_\infty$ the spectral norm of its observable, the explicit model uses $\mathcal{O}(D\log(D\norm{O}_{\infty}\delta^{-1}))$ additional qubits and gates.
\end{proposition}

The general idea behind this construction is to encode the input data $\bm{x}$ in ancilla qubits, to \emph{finite precision}, which can then be used repeatedly to approximate data-encoding gates using data-independent unitaries. More precisely, all data components $x_i\in\mathbb{R}$ of an input vector $\bm{x} = (x_1, \ldots, x_d)$ are encoded as bit-strings $\ket{\widetilde{\bm{x}}_i} = \ket{b_0b_1\ldots b_{p-1}} \in \{0,1\}^{p}$, to some precision $\varepsilon = 2^{-p}$ (e.g., using $R_x(b_j)$ rotations on $\ket{0}$ states).\break Now, using $p$ \emph{fixed} rotations, e.g., of the form $R_z(2^{-j})$, controlled by the bits $\ket{b_j}$ and acting on $n$ ``working'' qubits, one can encode every $x_i$ in arbitrary (multi-qubit) rotations $e^{-\mathrm{i} x_i H}$, e.g., $R_z(x_i)$, arbitrarily many times. Given that all these fixed rotations are data-independent, the feature encoding of any such circuit hence reduces to the encoding of the classical bit-strings $\widetilde{\bm{x}}_i$, prior to all variational operations.\break By preserving the variational unitaries appearing in a data re-uploading circuit and replacing its encoding gates with such controlled rotations, we can then approximate any data re-uploading model of the form of Eq.~(\ref{eq:dr-model}). The approximation error $\delta$ of this mapping originates from the finite precision $\varepsilon$ of encoding $\bm{x}$, which results in an imperfect implementation of the encoding gates in the original circuit. But as $\varepsilon \rightarrow 0$, we also have $\delta \rightarrow 0$, and the scaling of $\varepsilon$ (or the number of ancillas $dp$) as a function of $\delta$ is detailed in Appendix \ref{sec:dr-mappings}.

\subsubsection{Exact mapping}

We now move to our main construction, resulting in \emph{exact} mappings between data re-uploading and explicit models, i.e., that achieve $\delta = 0$ with finite resources. We rely here on a similar idea to our previous construction, in which we encode the input data on ancilla qubits and later use data-independent operations to implement the encoding gates on the working qubits. The difference here is that we use gate teleportation techniques, a form of measurement-based quantum computation~\cite{briegel09}, to directly implement the encoding gates on ancillary qubits and teleport them back (via entangled measurements) onto the working qubits when needed (see Fig.~\ref{fig:example-exact-mapping}).

\begin{thm}\label{thm:exact-mapping}
Given an arbitrary data re-uploading model $f_{\bm{\theta}}(\bm{x}) = \textnormal{Tr}[\rho_{\bm{\theta}}(\bm{x})O_{\bm{\theta}}]$ as specified by Eq.~(\ref{eq:dr-model}), there exists a mapping that produces an equivalent explicit model $\tilde{f}_{\bm{\theta}}(\bm{x}) = \textnormal{Tr}[\rho'(\bm{x})O'_{\bm{\theta}}]$ as specified by Eq.~(\ref{eq:explicit-model}), such that:
\begin{equation}\label{eq:explicit-data-reuploading}
\textnormal{Tr}[\rho'(\bm{x})O'_{\bm{\theta}}] = \textnormal{Tr}[\rho_{\bm{\theta}}(\bm{x})O_{\bm{\theta}}],\ \forall \bm{x}, \bm{\theta}.
\end{equation}
and $\norm{O'_{\bm{\theta}}}_\infty^2 \leq (1-\delta')^{-1}\norm{O_{\bm{\theta}}}_\infty^2$, for an arbitrary re-normalization parameter $\delta'>0$. For $D$ the number of encoding gates used by the data re-uploading model, the equivalent explicit model uses $\mathcal{O}(D\log(D/\delta'))$ additional qubits and gates.
\end{thm}

As we detail in Appendix \ref{sec:dr-mappings}, gate teleportation cannot succeed with unit probability without gate-dependent (and hence data-dependent) corrections conditioned on the measurement outcomes of the ancilla. But since we only care about equality in expectation values ($\textnormal{Tr}[\rho_{\bm{\theta}}(\bm{x})O_{\bm{\theta}}]$ and $\textnormal{Tr}[\rho'(\bm{x})O'_{\bm{\theta}}]$), we can simply discard these measurement outcomes in the observable $O'_{\bm{\theta}}$ (i.e., project on the correction-free measurement outcomes). In general, this leads to an observable with a spectral norm $\norm{O'_{\bm{\theta}}}_\infty^2 = 2^D \norm{O_{\bm{\theta}}}_\infty^2$ exponentially larger than originally, and hence a model that is exponentially harder to evaluate to the same precision. Using a \emph{nested} gate-teleportation scheme (see Appendix \ref{sec:dr-mappings}) with repeated applications of the encoding gates, we can however efficiently make this norm overhead arbitrarily small.

\begin{figure}[t]
\begin{center}
\includegraphics[width=\linewidth]{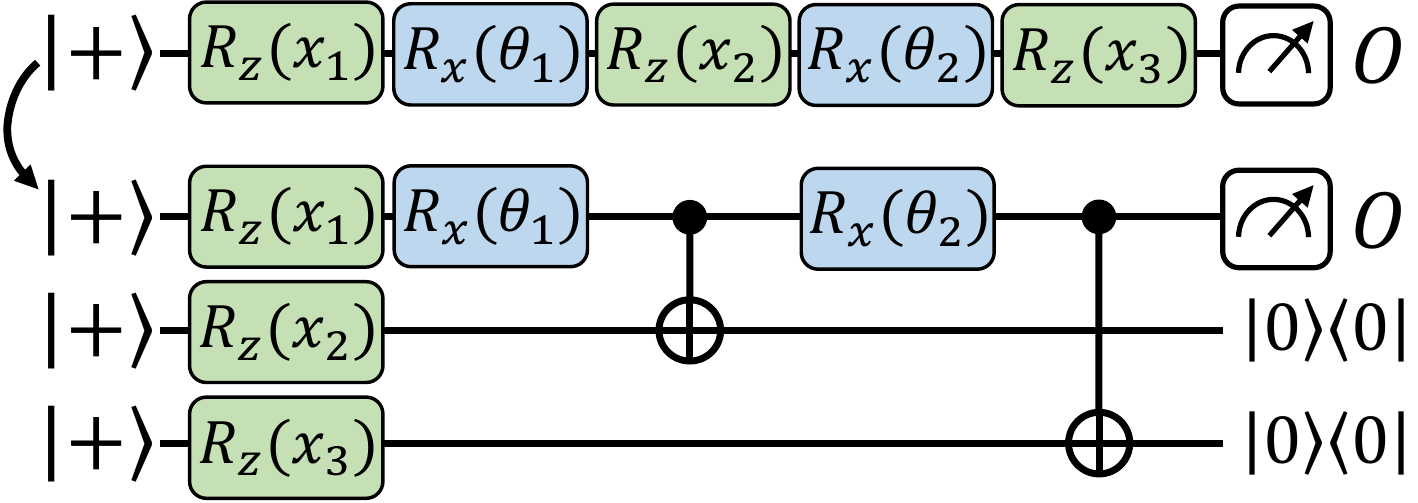}
\caption{An exact mapping from a data re-uploading model to an equivalent explicit model, using gate teleportation. The details of this mapping, as well as its more elaborate form (using nested gate teleportation) can found in Appendix \ref{sec:dr-mappings}.}
\label{fig:example-exact-mapping}
\end{center}
\vspace{-2em}
\end{figure}

As our findings indicate, mappings from data re-uploading to explicit models are not unique, and seem to always incur the use of additional qubits. In Sec.~\ref{sec:learning-sep} we prove that this is indeed the case, and that any mapping from an arbitrary data re-uploading model with $D$ encoding gates to an equivalent explicit model must use $\Omega(D)$ additional qubits in general. This makes our gate-teleportation mapping essentially optimal (i.e., up to logarithmic factors) in this extra cost.

To summarize, in this section, we demonstrated that linear quantum models can describe not only explicit and implicit models, but also data re-uploading circuits. More specifically, we showed that any hypothesis class of data re-uploading models can be mapped to an equivalent class of explicit models, that is, linear models with a restricted family of observables. In Appendix \ref{sec:explicit-universal}, we extend this result and show that explicit models can also approximate any \emph{computable} (classical or quantum) hypothesis class.

\section{Outperforming kernel methods with explicit and data re-uploading models}

From the standpoint of relating quantum models to each other, we have shown that the framework of linear quantum models allows to unify all standard models based on parametrized quantum circuits. While these findings are interesting from a theoretical perspective, they do not reveal how these models compare in practice. In particular, we would like to understand the advantages of using a certain model rather than the other in order to solve a given learning task. In this section, we address this question from several perspectives. First, we revisit the comparison between explicit and implicit models and clarify the implications of the representer theorem on the performance guarantees of these models. Then, we derive lower bounds for all three quantum models studied in this work in terms of their resource requirements, and show the existence of exponential separations between each pair of models. Finally, we discuss the implications of these results on the search for a quantum advantage in machine learning.

\subsection{Classical background and the representer theorem\label{sec:rkhs-representer}}

Interestingly, a piece of functional analysis from learning theory gives us a way of characterizing any family of linear quantum models \cite{schuld21}. Namely, the so-called \emph{reproducing kernel Hilbert space}, or RKHS \cite{scholkopf02}, is the Hilbert space $\mathcal{H}$ spanned by all functions of the form $f(\bm{x}) = \langle \phi(\bm{x}),w\rangle_{\mathcal{F}}$, for all $w \in \mathcal{F}$. It includes any explicit and implicit models defined by the quantum feature states $\phi(\bm{x}) = \rho(\bm{x})$. From this point of view, a relaxation of any learning task using implicit or explicit models as a hypothesis family consists in finding the function in the RKHS $\mathcal{H}$ that has optimal learning performance.
For the supervised learning task of modeling a target function $g(\bm{x})$ using a training set $\{(\bm{x}^{(1)}, g(\bm{x}^{(1)}), \ldots, (\bm{x}^{(M)}, g(\bm{x}^{(M)})\}$, this learning performance is usually measured in terms of a training loss of the form, e.g.,
\begin{equation}
	\widehat{\mathcal{L}}(f) = \frac{1}{M}\sum_{m=1}^{M} \left( f(\bm{x}^{(m)})-g(\bm{x}^{(m)}) \right)^2.
\end{equation}
The true figure of merit of this problem however, is in minimizing the \emph{expected} loss $\mathcal{L}(f)$, defined similarly as a probability-weighted average over the entire data space $\mathcal{X}$. For this reason, a so-called regularization term $\lambda \norm{f}^2_{\mathcal{H}} = \lambda \norm{O}^2_{\mathcal{F}}$ is often added to the training loss $\widehat{\mathcal{L}}_\lambda (f) = \widehat{\mathcal{L}}(f) + \lambda \norm{O}^2_{\mathcal{F}}$ to incentivize the model not to overfit on the training data. Here, $\lambda \geq 0$ is a hyperparameter that controls the strength of this regularization.

Learning theory also allows us to characterize the linear models in $\mathcal{H}$ that are optimal with respect to the regularized training loss $\widehat{\mathcal{L}}_\lambda(f)$, for any $\lambda \geq 0$. Specifically, the \emph{representer theorem} \cite{scholkopf02} states that the model $f_{\text{opt}}\in \mathcal{H}$ minimizing $\widehat{\mathcal{L}}_\lambda(f)$ is always a kernel model of the form of Eq.~(\ref{eq:kernel-model}) (see Appendix \ref{sec:representer-thm} for a formal statement).\break A direct corollary of this result is that implicit quantum models are guaranteed to achieve a lower (or equal) regularized training loss than any explicit quantum model using the same feature encoding \cite{schuld21}. Moreover, the optimal weights $\alpha_m$ of this model can be computed efficiently using $\mathcal{O}(M^2)$ evaluations of inner products on a quantum computer (that is, by estimating the expectation value in Fig.~\ref{fig:quantum-models}.b for all pairs of training points\footnote{For this work, we ignore the required precision of the estimations. We note however that these can require exponentially many measurements in the number of qubits, both for explicit \cite{mcclean18}\break and implicit \cite{kubler21} models.}) and with classical post-processing in time $\mathcal{O}(M^3)$ using, e.g., ridge regression or support vector machines \cite{scholkopf02}.

This result may be construed to suggest that, in our study of quantum machine learning models, we only need to worry about implicit models, where the only real question to ask is what feature encoding circuit we use to compute a kernel function, and all machine learning is otherwise classical.
In the next subsections, we show however the value of explicit and data re-uploading approaches in terms of generalization performance and resource requirements. 

\subsection{Explicit can outperform implicit models\label{sec:explicit-vs-implicit}}

We turn our attention back to the explicit models resulting from our approximate mappings (see Fig.~\ref{fig:construction}). Note that the kernel function associated to their bit-string encodings $\ket{\psi(\bm{x})} = \ket{0}^{\otimes n} \ket{\tilde{\bm{x}}}$, $\rho(\bm{x}) = \ket{\psi(\bm{x})}\!\!\bra{\psi(\bm{x})}$, is trivially
\begin{equation}\label{eq:res-kernel}
    k(\bm{x},\bm{x'}) = \prod_{i=1}^{d} \abs{\langle \widetilde{\bm{x}}_i | \widetilde{\bm{x}}'_i\rangle}^2 = \delta_{\widetilde{\bm{x}}, \widetilde{\bm{x}}'},
\end{equation}
that is, the Kronecker delta function of the bit-strings $\widetilde{\bm{x}}$ and $\widetilde{\bm{x}}'$. Let us emphasize that, for an appropriate precision $\varepsilon$ of encoding input vectors $\bm{x}$, the family of explicit models resulting from our construction includes good approximations of \emph{virtually any} parametrized quantum circuit model acting on $n$ qubits. Yet, all of these result in the same kernel function of Eq.~(\ref{eq:res-kernel}).
This is a rather surprising result, for two reasons. First, this kernel is classically computable, which, in light of the representer theorem, seems to suggest that a simple classical model of the form of Eq.~(\ref{eq:kernel-model}) can outperform \emph{any} explicit quantum model stemming from our construction, and hence any quantum model in the limit $\varepsilon \rightarrow 0$. Second, this implicit model always takes the form
\begin{equation}\label{eq:implicit-bitstring}
    f_{\bm{\alpha},\mathcal{D}}(\bm{x}) = \sum_{m=1}^{M} \alpha_m \delta_{\widetilde{\bm{x}}, \widetilde{\bm{x}}^{(m)}},
\end{equation}
which is a model that overfits the training data and fails to generalize to unseen data points, as, for $\varepsilon \rightarrow 0$ and any choice of $\bm{\alpha}$, $f_{\bm{\alpha},\mathcal{D}}(\bm{x}) = 0$ for any $\bm{x}$ outside the training set. As we detail in Appendix \ref{sec:dr-mappings}, similar observations can be made for the kernels resulting from our gate-teleportation construction.

These last remarks force us to rethink our interpretation of the representer theorem. When restricting our attention to the regularized training loss, implicit models do indeed lead to better training performance due to their increased expressivity.\footnote{On a classification task with labels $g(\bm{x}) = \pm 1$, the kernel model of Eq.~(\ref{eq:implicit-bitstring}) is optimal with respect to any regularized training loss for $\alpha_m=g(\bm{x}^{(m)})\ \forall m$ such that $\protect\widehat{\mathcal{L}}(f) = 0$ and  $\norm{f}_{\mathcal{H}}^2=M$.} But, as our construction shows, this expressivity can dramatically harm the generalization performance of the learning model, despite the use of regularization during training. Hence, restricting the set of observables accessible to a linear quantum model (or, equivalently, restricting the accessible manifold of the RKHS) can potentially provide a substantial learning advantage.

\subsection{Rigorous learning separations between all quantum models\label{sec:learning-sep}}

Motivated by the previous illustrative example, we analyze more rigorously the advantages of explicit and data re-uploading models over implicit models. For this, we take a similar approach to recent works in classical machine learning which showed that neural networks can efficiently solve some learning tasks that linear or kernel methods cannot \cite{daniely20,hsu21}. In our case, we quantify the efficiency of a quantum model in solving a learning task by the number of qubits and the size of the training set it requires to achieve a non-trivial expected loss. To obtain scaling separations, we consider a learning task specified by an arbitrary input dimension $d \in \mathbb{N}$ and express the resource requirements of the different quantum models as a function of $d$.

Similarly to Ref.~\cite{daniely20}, the learning task we focus on is that of learning parity functions (see Fig.~\ref{fig:summary-results_2}). These functions take as input a $d$-dimensional binary input $\bm{x}\in\{-1,1\}^d$ and return the parity (i.e., the product) of a certain subset $A \subset \{1,\ldots,d\}$ of the components of $\bm{x}$.\break The interesting property of these functions is that, for any two choices of $A$, the resulting parity functions are orthogonal in the Hilbert space $\mathcal{H}$ of functions from $\{-1,1\}^d$ to $\mathbb{R}$. Hence, since the number of possible choices for $A$ grow combinatorially with $d$, the subspace of $\mathcal{H}$ that these functions span also grows combinatorially with $d$ (can be made into a $2^d$ scaling by restricting the choices of $A$). On the other hand, a linear model (explicit or implicit) also covers a restricted subspace (or manifold) of $\mathcal{H}$. The dimension of this subspace is upper bounded by $2^{2n}$ for a quantum linear model acting on $n$ qubits, and by $M$ for an implicit model using $M$ training samples (see Appendix \ref{sec:learn-sep} for detailed explanations). Hence by essentially comparing these dimensions ($2^d$ versus $2^{2n}$ and $M$) \cite{hsu21}, we can derive our lower bounds for explicit and implicit models. As for data re-uploading models, they do not suffer from these dimensionality arguments. The different components of $\bm{x}$ can be processed sequentially by the model, such that a single-qubit data re-uploading quantum circuit can represent (and learn) any parity function.

We summarize our results in the following theorem, and refer to Appendix \ref{sec:learn-sep} for a more detailed exposition.

\begin{figure}[t]
\includegraphics[width=0.85\linewidth]{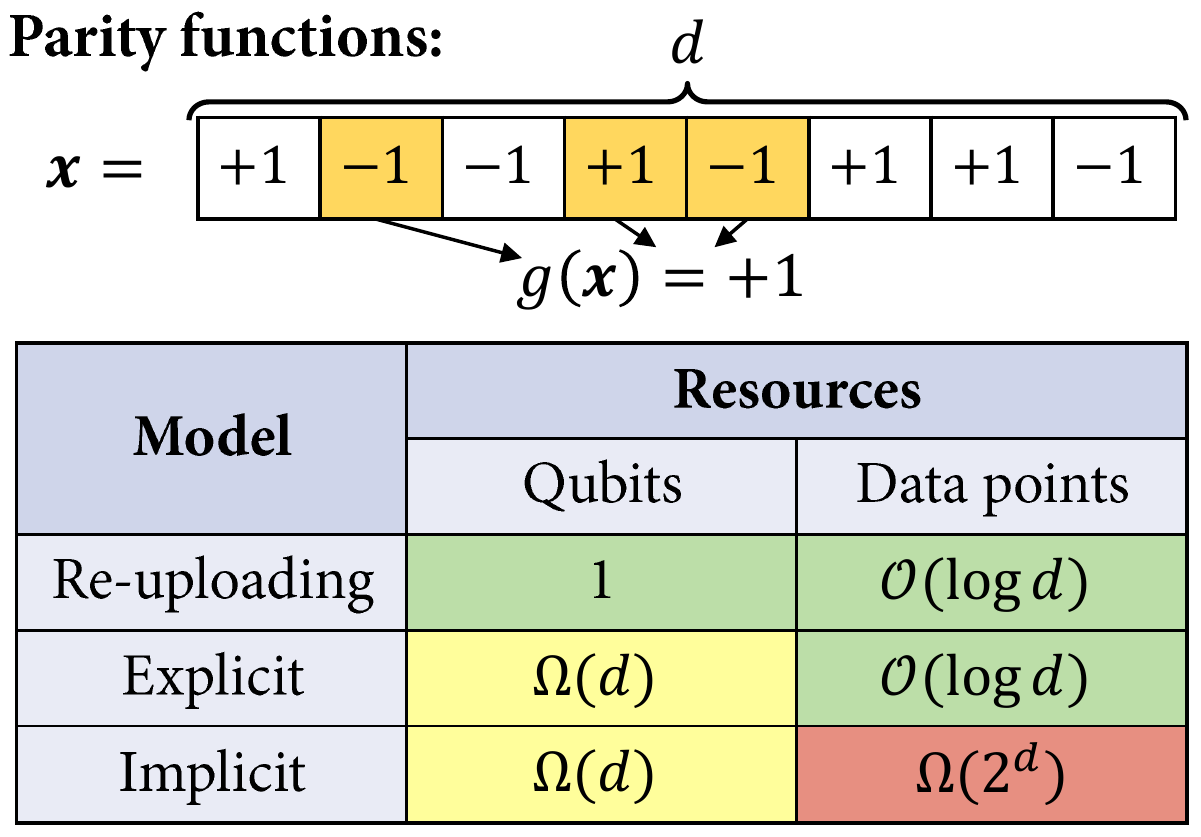}
\caption{Learning separations. We describe a learning task based on parity functions acting on $d$-bit input vectors $\bm{x} \in \{-1,1\}^d$, for $d\in\mathbb{N}$. This task allows us to separate all three quantum models studied in this work in terms of their resource requirements, as a function of $d$ (see Theorem \ref{thm:separation-main}).}
\label{fig:summary-results_2}
\vspace{-1em}
\end{figure}

\begin{thm}\label{thm:separation-main}
There exists a regression task specified by an input dimension $d \in \mathbb{N}$, a function family $\{g_A : \{-1,1\}^d \rightarrow \{-1,1\}\}_{A}$, and associated input distributions $\mathcal{D}_{A}$, such that, to achieve an average mean-squared error
$$\mathbb{E}_{A} \big[ \inf_{f} \norm{f-g_A}^2_{L^2(\mathcal{D}_{A})} \big] = \varepsilon < 1/2$$
(i) any linear quantum model needs to act on
$$n \geq \Omega(d + \log(1-2\varepsilon))$$
qubits,\\
(ii) any implicit quantum model additionally requires
$$M \geq \Omega(2^{d}(1-2\varepsilon))$$
data samples, while\\
(iii) a data re-uploading model acting on a single qubit and using $d$ encoding gates can be trained to achieve a perfect expected error with probability $1-\delta$, using $M = \mathcal{O}(\log(\frac{d}{\delta}))$ data samples.
\end{thm}

A direct corollary of this result is a lower bound on the number of additional qubits that a universal mapping from any data re-uploading model to  equivalent explicit models must use:

\begin{corollary}\label{cor:lower-bound}
Any universal mapping that takes as input an arbitrary data re-uploading model $f_{\bm{\theta}}$ with $D$ encoding gates and maps it to an equivalent explicit model $\widetilde{f}_{\bm{\theta}}$ must produce models acting on $\Omega(D)$ additional qubits for worst-case inputs.
\end{corollary}

Comparing this lower bound to the scaling of our gate-teleportation mapping (Theorem \ref{thm:exact-mapping}), we find that it is optimal up to logarithmic factors.

\subsection{Quantum advantage beyond kernel methods\label{sec:num-comp}}

A major challenge in quantum machine learning is showing that the quantum methods discussed in this work can achieve a learning advantage over (standard) classical methods. While some approaches to this problem focus on constructing learning tasks with separations based on complexity-theoretic assumptions \cite{liu20,sweke20}, other works try to assess empirically the type of learning problems where quantum models show an advantage over standard classical models \cite{schuld20,huang20}. In this line of research, Huang \emph{et al.}~\cite{huang20} propose looking into learning tasks where the target functions are themselves generated by (explicit) quantum models. Following similar observations to those made in Sec.~\ref{sec:rkhs-representer} about the learning guarantees of kernel methods, the authors also choose to assess the presence of quantum advantages by comparing the learning performance of standard classical models only to that of implicit quantum models (from the same family as the target explicit models).\break This restricted comparison led to the conclusion that, with the help of training data, classical machine learning models could be as powerful as quantum machine learning models, even in these tailored learning tasks.

Having discussed the limitations of kernel methods in the previous subsections, we revisit this type of numerical experiments, where we additionally evaluate the performance of explicit models on these type of tasks.

\begin{figure}[t]
\includegraphics[width=\linewidth]{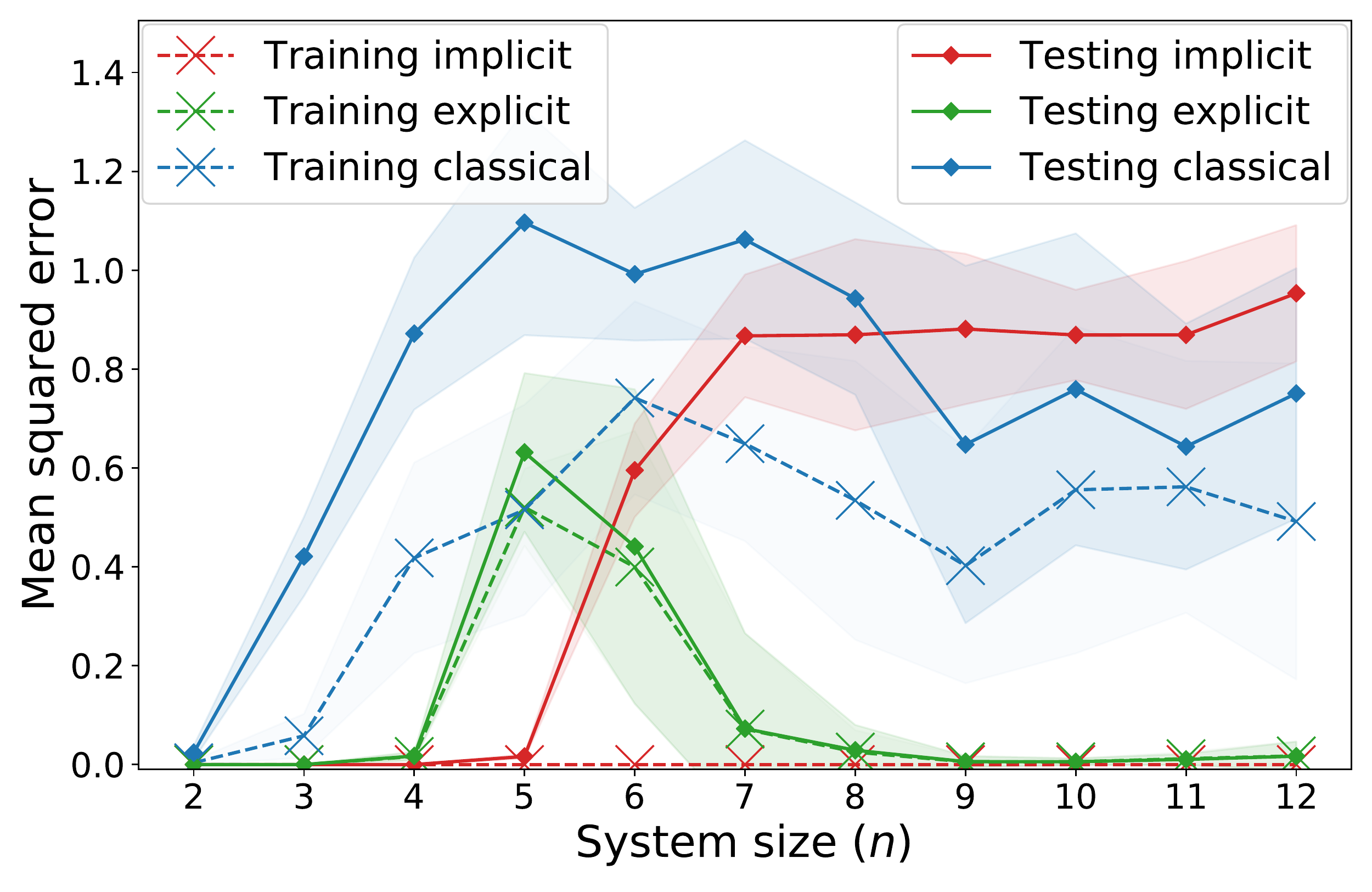}
\vspace{-2em}
\caption{Regression performance of explicit, implicit and classical models on a ``quantum-tailored'' learning task. For all system sizes, each model has access to a training set of $M=1000$ pre-processed and re-labeled fashion-MNIST images.\break Testing loss is computed on a test set of size $100$. Shaded regions indicate the standard deviation over $10$ labeling functions. The training errors of implicit models are close to $0$ for all system sizes.}
\vspace{-1em}
\label{fig:fashion_regression}
\end{figure}

Similarly to Huang \emph{et al.}~\cite{huang20}, we consider a regression task with input data from the fashion-MNIST dataset \cite{xiao17}, composed of 28x28-pixel images of clothing items. Using principal component analysis, we first reduce the dimension of these images to obtain $n$-dimensional vectors, for $2 \leq n \leq 12$. We then label the images using an explicit model acting on $n$ qubits. For this, we use the feature encoding proposed by Havl\'{i}\v{c}ek \emph{et al.} \cite{havlivcek19}, which is conjectured to lead to classically intractable kernels, followed by a hardware-efficient variational unitary \cite{peruzzo14}. The expectation value of a Pauli $Z$ observable on the first qubit then produces the data labels.\footnote{Note that we additionally normalize the labels as to obtain a standard deviation of $1$ for all system sizes.} On this newly defined learning task, we test the performance of explicit models from the same function family as the explicit models generating the (training and test) data, and compare it to that of implicit models using the same feature encoding (hence from the same extended family of linear models), as well as a list of standard classical machine learning algorithms that are hyperparametrized for the task (see Appendix \ref{sec:num-sim}). The results of this experiment are presented in Fig.\ \ref{fig:fashion_regression}.

The training losses we observe are consistent with our previous findings: the implicit models systematically achieve a lower training loss than their explicit counterparts (for an unregularized loss\footnote{The addition of regularization to the training loss of the implicit model does not impact the separation we observe here (see Appendix \ref{sec:add-num}).} notably, the implicit models achieve a training loss of $0$). With respect to the testing loss on the other hand, which is representative of the expected loss, we see a clear separation starting from $n=7$ qubits, where the classical models start having a competitive performance with the implicit models, while the explicit models clearly outperform them both. This goes to show that the existence of a quantum advantage should not be assessed only by comparing classical models to quantum kernel methods, as explicit (or data re-uploading) models can also conceal a substantially better learning performance.

\vspace{-2.5em}

\section{Conclusion}

In this work, we present a unifying framework for quantum machine learning models by expressing them as linear models in quantum feature spaces. In particular, we show how data re-uploading circuits can be represented exactly by explicit linear models in larger feature spaces. While this unifying formulation as linear models may suggest that all quantum machine learning models should be treated as kernel methods, we illustrate the advantages of variational quantum methods for machine learning. Going beyond the advantages in training performance guaranteed by the representer theorem, we first show how a systematic ``kernelization" of linear quantum models can be harmful in terms of their generalization performance. Further, we analyze the resource requirements (number of qubits and data samples used by) these models, and show the existence of exponential separations between data re-uploading, linear, and kernel quantum models to solve certain learning tasks.

One take-away message from our results is that training loss, even when regularized, is a misleading figure of merit. Generalization performance, which is measured on seen as well as unseen data, is in fact the important quantity to care about in (quantum) machine learning. These two sentences written outside of context will seem obvious to individuals well-versed in learning theory. However, it is crucial to recall this fact when evaluating the consequences of the representer theorem. This theorem only discusses regularized training loss, and thus despite its guarantees on the training loss of quantum kernel methods, it allows explicit models to have an exponential learning advantage in the number of data samples they use to achieve a good generalization performance.

From the limitations of quantum kernel methods highlighted by these results, we revisit a discussion on the power of quantum learning models relative to classical models in machine learning tasks with quantum-generated data. In a similar learning task to that of Huang \emph{et al.}~\cite{huang20}, we show that, while standard classical models can be competitive with quantum kernel methods even in these ``quantum-tailored'' problems, variational quantum models can exhibit a significant learning advantage. These results give us a more comprehensive view of the quantum machine learning landscape and broaden our perspective on the type of models to use in order to achieve a practical learning advantage in the NISQ regime.

\section{Discussion}

In this paper, we focus on the theoretical foundations of quantum machine learning models and how expressivity impacts generalization performance. But a major practical consideration is also that of \emph{trainability} of these models. In fact, we know of obstacles in trainability for both explicit and implicit models. Explicit models can suffer from barren plateaus in their loss landscapes \cite{mcclean18,cerezo21}, which manifest in exponentially vanishing gradients in the number of qubits used, while implicit models can suffer from exponentially vanishing kernel values \cite{kubler21,thanasilp22}. While these phenomena can happen under different conditions, they both mean that an exponential number of circuit evaluations can be needed to train and make use of these models. Therefore, aside from the considerations made in this work, emphasis should also be placed on avoiding these obstacles to make good use of quantum machine learning models in practice.

The learning task we consider to show the existence of exponential learning separations between the different quantum models is based on parity functions, which is not a concept class of practical interest in machine learning. We note however that our lower bound results can also be extended to other learning tasks with concept classes of large dimension (i.e., composed of many orthogonal functions). Quantum kernel methods will necessarily need a number of data points that scales linearly with this dimension, while, as we showcased in our results, the flexibility of data re-uploading circuits, as well as the restricted expressivity of explicit models can lead to substantial savings in resources. It remains an interesting research direction to explore how and when can these models be tailored to a machine learning task at hand, e.g., through the form of useful \emph{inductive biases} (i.e., assumptions on the nature of the target functions) in their design.

\section*{Code availability}
The code used to run the numerical simulations, implemented using TensorFlow Quantum \cite{broughton20}, is available at \href{https://github.com/sjerbi/QML-beyond-kernel}{https://github.com/sjerbi/QML-beyond-kernel} \cite{jerbi23}.

\section*{Acknowledgments}
The authors would like to thank Isaac D.~Smith, Casper Gyurik, Matthias C.~Caro, Elies Gil-Fuster, Ryan Sweke, and Maria Schuld for helpful discussions and comments, as well as Hsin-Yuan Huang for clarifications on their numerical simulations \cite{huang20}. SJ, LJF, HPN and HJB acknowledge support from the Austrian Science Fund (FWF) through the projects DK-ALM:W1259-N27 and SFB BeyondC F7102. SJ also acknowledges the Austrian Academy of Sciences as a recipient of the DOC Fellowship. HJB also acknowledges support by the European Research Council (ERC) under Project No. 101055129. HJB was also supported by the Volkswagen Foundation (Az:97721). This work was in part supported by the Dutch Research Council (NWO/OCW), as part of the Quantum Software Consortium program (project number 024.003.037). VD acknowledges the support by the project NEASQC funded from the European Union’s Horizon 2020 research and innovation programme (grant agreement No 951821). VD also acknowledges support through an unrestricted gift from Google Quantum AI.

\bibliography{references}

\clearpage

\appendix

\section{Representer theorem\label{sec:representer-thm}}

In this appendix, we give a formal statement of the representer theorem from learning theory.

\begin{theorem}[Representer theorem \cite{scholkopf02}]
Let $g: \mathcal{X} \rightarrow \mathcal{Y}$ be a target function with input and output domains $\mathcal{X}$ and $\mathcal{Y}$, $\mathcal{D} = \{(\bm{x}^{(1)}, g(\bm{x}^{(1)}), \ldots, (\bm{x}^{(M)}, g(\bm{x}^{(M)})\}$  a training set of size $M$, and $k: \mathcal{X}\times\mathcal{X} \rightarrow \mathbb{R}$ a kernel function with a corresponding reproducing kernel Hilbert space (RKHS) $\mathcal{H}$. For any strictly monotonic increasing regularization function $h:[0,\infty) \rightarrow \mathbb{R}$ and any training loss $\widehat{\mathcal{L}}: (\mathcal{X}\times\mathcal{Y})^M\times\mathcal{Y}^M \rightarrow \mathbb{R}\cup\infty$, we have that any minimizer of the regularized training loss from the RKHS $\mathcal{H}$, 
\begin{equation*}
	f_{\textnormal{opt}} = \underset{f\in\mathcal{H}}{\textnormal{argmin}} \left\{ \widehat{\mathcal{L}}(\mathcal{D},f(\mathcal{D}_{\bm{x}})) + h(\norm{f}^2_{\mathcal{H}}) \right\} 
\end{equation*}
admits a representation of the form
\begin{equation*}
	f_{\textnormal{opt}}(\bm{x}) = \sum_{m=1}^{M} \alpha_m k(\bm{x},\bm{x}^{(m)})
\end{equation*}
where $\alpha_m \in \mathbb{R}$ for all $1\leq m \leq M$.
\end{theorem}

A common choice for the regularization function is simply $h(\norm{f}^2_{\mathcal{H}}) = \lambda \norm{f}^2_{\mathcal{H}}$, where $\lambda \geq 0$ is a hyperparameter adjusting the strength of the regularization.

\section{Mappings from data re-uploading to explicit models\label{sec:dr-mappings}}

In this section, we detail possible mappings from data re-uploading models to explicit models, and prove Proposition \ref{prop:approx-mapping} and Theorem \ref{thm:exact-mapping} from Sec.~\ref{sec:linear-reup}.

In our analysis, we restrict our attention to encoding gates of the form $e^{-\mathrm{i} h(\bm{x}) H_n / 2}$, for $H_n$ an arbitrary Pauli string acting on $n$ qubits, e.g., $H_3 = X \otimes Z \otimes I$, and $h: \mathbb{R}^d \rightarrow \mathbb{R}$ an arbitrary function mapping real-valued input vectors $\bm{x}$ to rotation angles $h(\bm{x})$. Using known techniques (see Sec.\ 4.7.3 in \cite{nielsen02}), we can show that in order to implement any such gate $e^{-\mathrm{i} h(\bm{x}) H_n/2}$ \emph{exactly}, one only needs to perform a Pauli-Z rotation $e^{-\mathrm{i} h(\bm{x}) Z/2}$ on a single of these $n$ qubits, along with $\mathcal{O}(n)$ operations that are independent of $\bm{x}$ (and can therefore be absorbed by the surrounding variational unitaries). Therefore, in our mappings, we only need to focus on encoding gates of the form $R_z(h(\bm{x})) = e^{-\mathrm{i} h(\bm{x}) Z/2}$.

\subsection{Approximate bit-string mapping}

We start by analyzing the resource requirements (in terms of number of additional qubits and gates) of our approximate bit-string mapping (Proposition \ref{prop:approx-mapping}).

Note that, in our construction (see Fig.~\ref{fig:construction}), the ancilla qubits are always prepared in computational basis states and only act as classical controls throughout the circuit. Hence, the operation $\textnormal{Tr}_{\widetilde{\bm{x}}}[V_1(\bm{\theta})C\text{-}\widetilde{U}_1\!\ldots\! V_D(\bm{\theta})C\text{-}\widetilde{U}_D]$ obtained by tracing out these ancillas is equivalent to the \emph{unitary} $V_1(\bm{\theta})U_1(\widetilde{x}_1)\!\ldots\!V_D(\bm{\theta})U_D(\widetilde{x}_D)$, where data-dependent rotations are only implemented to angle-precision $\varepsilon=2^{-p}$. In the following, we relate this precision $\varepsilon$ (or equivalently, the number of ancilla qubits $dp$)\footnote{We assume the worst-case scenario where each component $x_i$ of $\bm{x}$ is assigned to a unique encoding gate, such that $D=d$.} to the approximation error $\delta\geq|\tilde{f}_{\bm{\theta}}(\bm{x}) - f_{\bm{\theta}}(\bm{x})|$ of our mapping.

Call $U = V_1(\bm{\theta})U_1(x_1)\!\ldots\! V_D(\bm{\theta})U_D(x_D)$ the data re-uploading unitary and $V= V_1(\bm{\theta})U_1(\widetilde{x}_1) \!\ldots\! V_D(\bm{\theta})U_D(\widetilde{x}_D)$ its approximation. We first relate the error $\delta$ to the distance measure $\norm{U-V}_\infty=\max_{\ket{\psi}}\!\norm{(U\!-\!V)\ket{\psi}}$:
\begin{align}
	\abs{\tilde{f}_{\bm{\theta}}(\bm{x}) - f_{\bm{\theta}}(\bm{x})} &= \abs{\bra{\psi}U^{\dagger}OU\ket{\psi} - \bra{\psi}V^{\dagger}OV\ket{\psi}} \nonumber\\
	&= \abs{\bra{\psi}U^{\dagger}O\ket{\Delta} - \bra{\Delta}OV\ket{\psi}} \nonumber\\
	&\leq \abs{\bra{\psi}U^{\dagger}O\ket{\Delta}} + \abs{\bra{\Delta}OV\ket{\psi}} \nonumber\\
	&\leq \begin{multlined}[t][5cm]\sqrt{\abs{\bra{\psi}U^{\dagger}O^2U\ket{\psi}}}\norm{\ket{\Delta}} \\+ \sqrt{\abs{\bra{\psi}V^{\dagger}O^2V\ket{\psi}}}\norm{\ket{\Delta}}\end{multlined} \nonumber\\
	&\leq 2 \norm{O}_\infty \norm{U-V}_\infty
\end{align}
for $\ket{\Delta} = (U-V)\ket{\psi}$, by using the triangular and Cauchy-Schwarz inequalities to derive the first two inequalities, and the definition of the spectral norm $\norm{O}_{\infty}$ and the distance $\norm{U-V}_\infty$ to derive the last inequality.

Note that, for $U$ and $V$ obtained by sequences of unitary gates, the distance $\norm{U-V}_\infty$ is linear in the pairwise distances between the gates in these sequences (see Sec.~4.5.3 of \cite{nielsen02}):
\begin{equation}
\norm{U-V}_\infty \leq \sum_{j=1}^{D} \norm{U_j-V_j}_\infty
\end{equation}
therefore, to obtain $|\tilde{f}_{\bm{\theta}}(\bm{x}) - f_{\bm{\theta}}(\bm{x})| \leq \delta$, it is sufficient to enforce:
\begin{equation}\label{eq:error-delta}
\norm{U_j-V_j}_\infty \leq \frac{\delta}{2D\norm{O}_\infty}
\end{equation}
for each and single encoding gate in the circuit.

Since we assumed that the encoding gates of the circuit take the form $U_j = R_z(x_j)$, we can bound $\norm{U_j-V_j}_\infty$ as a function of the precision $\varepsilon$ of encoding $x_j$ as:
\begin{align}
	\norm{U_j-V_j}_\infty &= \max_{\ket{\psi}}\!\norm{(R_z(x_j+\varepsilon)-R_z(x_j))\ket{\psi}} \nonumber\\
	&= \norm{e^{ix_j}(1-e^{i\varepsilon})\ket{1}} = \abs{1-e^{i\varepsilon}} \nonumber\\
	&\leq \sqrt{2}\varepsilon \label{eq:error-rot}
\end{align}

From Eqs.~(\ref{eq:error-delta}) and (\ref{eq:error-rot}), we then get $\varepsilon \leq \frac{\delta}{2\sqrt{2}D\norm{O}_\infty}$, or equivalently, $p \geq \lceil\log_2(\frac{2\sqrt{2}D\norm{O}_\infty}{\delta})\rceil$. The number of additional qubits in the circuit is then $Dp$, which is also the number of additional gates ($R_x$ data-encoding rotations and controlled data-independent rotations).

\subsection{Mappings based on gate teleportation}

We now move to our gate-teleportation mappings (Theorem \ref{thm:exact-mapping}). Again, we restrict our attention to teleporting $R_z(x)$ gates. This gate teleportation can easily be implemented using the gadget depicted in Fig.~\ref{fig:gadget}. It is easy to check that for an arbitrary input qubit $\ket{\psi} = \alpha \ket{0} + \beta \ket{1}$ and $\ket{+} = \frac{1}{\sqrt{2}}(\ket{0} + \ket{1})$, the state generated by this gadget before the computational basis measurement (and correction) is:
\begin{multline}\label{eq:gate-teleportation-gadget}
    \frac{1}{\sqrt{2}}(\alpha\ket{0} + \beta e^{\mathrm{i} x}\ket{1})\otimes \ket{0} + \frac{1}{\sqrt{2}}(\alpha e^{\mathrm{i} x}\ket{0} + \beta\ket{1})\otimes \ket{1} \\
    = \frac{1}{\sqrt{2}} \left( R_z(x)\ket{\psi}\otimes\ket{0} + e^{\mathrm{i} x} R_z(-x)\ket{\psi}\otimes\ket{1} \right)
\end{multline}
which results in the correct outcome $\ket{\psi'} = R_z(x) \ket{\psi}$ for a $\ket{0}$ measurement, and a state that can be corrected (up to a global phase) via a $R_z(2x)$ rotation, in case of  a $\ket{1}$ measurement.

\begin{figure}[t]
\begin{center}
\includegraphics[width=0.7\linewidth]{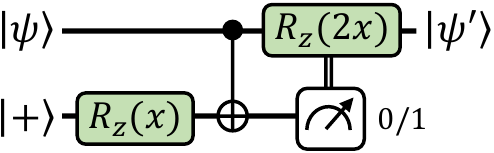}
\vspace{-0.5em}
\caption{A gate-teleportation gadget (inspired by the T-gadget in \cite{bravyi16}) that implements the unitary $\ket{\psi} \mapsto \ket{\psi'} = R_z(x) \ket{\psi}$.}
\label{fig:gadget}
\end{center}
\end{figure}

Putting aside the corrections required by this gadget, we note the interesting property that, when used to simulate every encoding gate in the data re-uploading circuit, this gadget moves all data-dependent parts of the circuit on additional ancilla qubits, essentially turning it into an explicit model. However, this gadget still requires data-dependent corrections (of the form $R_z(2h(\bm{x}))$) in the case of $\ket{1}$ measurement outcomes, which happen with probability $1/2$ for each gate teleportation. To get around this problem, we simply replace the computational basis measurement in the gadget by a projection $P_0 = \ket{0}\!\!\bra{0}$ on the $\ket{0}$ state. While these projections cannot be implemented deterministically in practice, the resulting model is still a valid explicit model, in the sense that including the projections $P_0^{\otimes D}$ in the observable $O_{\bm{\theta}}$, for $D$ uses of our gadget, still leads to a valid observable $O'_{\bm{\theta}}= O_{\bm{\theta}} \otimes P_0^{\otimes D}$.

However, given that each of these projections does not account for the re-normalization of the resulting quantum state (i.e., the factor $1/\sqrt{2}$ in Eq.~(\ref{eq:gate-teleportation-gadget})), this means that, in order to enforce $\text{Tr}[\rho'(\bm{x})O'_{\bm{\theta}}] = \text{Tr}[\rho_{\bm{\theta}}(\bm{x})O_{\bm{\theta}}],\ \forall \bm{x}, \bm{\theta}$, we need to multiply  $O'_{\bm{\theta}}$ by a factor of $2^{D/2}$. This implies that the evaluation of the resulting explicit model is exponentially harder than that of the original data re-uploading model, in the number of encoding gates $D$. As we show next, this factor can however be made arbitrarily close to $1$, by allowing each encoding gate/angle to be used more than once in the feature encoding.

To achieve this feat, we transform our previous gadget as to implement its data-dependent corrections using gate teleportation again. A single such nested use of gate teleportation now has probability $1-1/4 = 3/4$ of succeeding without corrections, as opposed to the probability $1/2$ of the previous gadget. For $N$ nested uses (see Fig.~\ref{fig:gadget-2}), the success probability is then boosted to $1-2^{-N}$, which can be made arbitrarily close to $1$. If we use this nested gadget for all $D$ encoding gates in the circuit, the probability that all of them are implemented successfully without corrections is then $p = (1-2^{-N})^D$. This probability $p$ can be made larger than $1-\delta'$, for any $\delta' > 0$, by choosing $N = \left\lceil \log_2 \left((1-\sqrt[D]{1-\delta'})^{-1}\right) \right\rceil \leq \left\lceil\log_2(D/\delta')\right\rceil$. This also leads to a normalization factor $\sqrt{p^{-1}} \leq \sqrt{(1-\delta')^{-1}}$, which can be made arbitrarily close to $1$. Note as well that this normalization factor is always known exactly, such that we can guarantee $\text{Tr}[\rho'(\bm{x})O'_{\bm{\theta}}] = \text{Tr}[\rho(\bm{x},\bm{\theta})O_{\bm{\theta}}],\ \forall \bm{x}, \bm{\theta}$.

\begin{figure}[t]
\begin{center}
\includegraphics[width=0.75\linewidth]{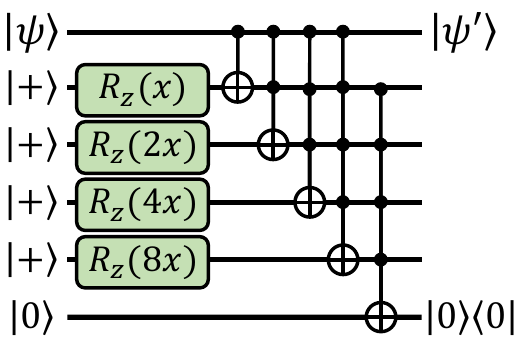}
\caption{A gate-teleportation gadget without data-dependent corrections and success probability $1-2^{-4} \approx 0.94$. The last qubit acts as a witness that at least one of the nested gate teleportations succeeded without the need for corrections.}
\label{fig:gadget-2}
\end{center}
\vspace{-2em}
\end{figure}

\subsection{Kernels resulting from our mappings}

In the main text, we showed how our illustrative mapping based on bit-string encodings of $\bm{x}$ resulted in trivial Kronecker-delta kernel functions and implicit models with very poor generalization performance. In this subsection, we derive a similar result for our gate-teleportation mappings.

We first note that these mappings lead to feature encodings of the form:
\begin{equation*}
    U_\phi(\bm{x}) \ket{0^{\otimes n+ND+D}} = \ket{0^{\otimes n+D}}\bigotimes_{\substack{1\leq i\leq  D\\ 1\leq j \leq N}} R_z(2^{j-1} h_i(\bm{x})) \ket{+}
\end{equation*}
for $D$ encoding gates with encoding angles $h_i(\bm{x})$, and using $N$ nested gate teleportations for each of these gates. While less generic than the feature states resulting from our binary encodings, these still generate kernels
\begin{equation}
    k(\bm{x}, \bm{x'}) = \prod_{i,j}\cos(2^{j-1}(h_i(\bm{x}) -h_i(\bm{x'})))^2
\end{equation}
that are again classically simulatable and $k(\bm{x}, \bm{x'}) \rightarrow \delta_{\bm{x}, \bm{x'}}$ for $ND \rightarrow \infty$.

Moreover, in the case where the angles $h_i(\bm{x})$ are linear functions of components $x_i$ of $\bm{x}$, we can directly apply Theorem 1 of \cite{kubler21} to show the following. For a number of encoding gates $D$ and a number of nested gate teleportations $N$ large enough (i.e., $ND$ larger than some $d_0$), and for a dataset that is at most polynomially large in $ND$, no function can be learned using the implicit model resulting from this kernel. Note that, for this theorem to be applicable, we also need to assume non-degenerate data distributions $\mu$ (i.e., that do not have support on single data points) that are separable on all components $x_i$ of $\bm{x}$, i.e., $\mu = \bigotimes_i \mu_i$, such that the mean embeddings $\rho_{\mu_i} = \int \rho_i(\bm{x}) \mu_i(d\bm{x})$ for each component $x_i$ are all mixed.

\subsection{Link to no-programming}
It may seem to the informed reader as though our mappings from data re-uploading to explicit models are in violation of the so-called no-programming theorem from quantum information theory. In this section, we will shortly outline this theorem and explain why our mappings do not violate it.

A programmable quantum processor is defined as a CPTP map  $\mathcal{C}:\mathcal{H}_\text{S}\otimes \mathcal{H}_\text{P}\rightarrow \mathcal{H}_\text{S}$, where $\mathcal{H}_\text{S}$ and $\mathcal{H}_\text{P}$ denote the system and program Hilbert spaces. The purpose of such a quantum processor is to implement unitary maps $U:\mathcal{H}_\text{S}\rightarrow \mathcal{H}_\text{S}$ where the information about $U$ is fed to the processor solely by a program state $\ket{\psi_\text{P}}$ (see Fig.~\ref{fig:processor}).
The no-programming theorem \cite{nielsen97, yang20} rules out the existence of perfect universal quantum processors, in the sense that there cannot exist a processor $\mathcal{C}$ that can implement infinitely many different unitary maps $U$ deterministically using finite-dimensional program states $\ket{\psi_{\text{P}}}$.

\begin{figure}[t]
\begin{center}
\includegraphics[width=0.7\linewidth]{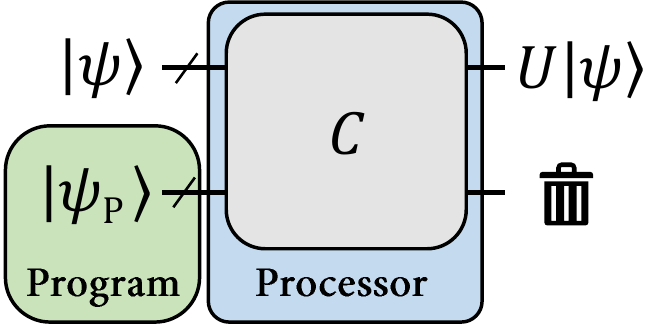}
\caption{A programmable quantum processor. A state $\ket{\psi_\text{P}}$ is fed to the processor as a program instructing the processor to implement a unitary map $U$ on another input state. The no-programming theorem states that programmable processors $\mathcal{C}$ capable of implementing any unitary map $U$ cannot exist.}
\label{fig:processor}
\end{center}
\vspace{-2em}
\end{figure}

The explicit models resulting from our mappings have properties that are reminiscent of quantum processors. In Fig.~\ref{fig:construction} for instance, the bit-string encodings $\ket{\widetilde{\bm{x}}}$ are used to implement data-encoding unitaries in an otherwise data-independent circuit. Thus, one may interpret these as the program states $\ket{\psi(\bm{x})}$ of a quantum processor $\mathcal{C}$ given by the rest of the circuit. Same goes for the quantum states $R_z(x_2)\ket{+}R_z(x_3)\ket{+}$ in Fig.~\ref{fig:example-exact-mapping}.

In light of the no-programming theorem, it is quite remarkable that these explicit models can ``program" a \textit{continuous} set of unitaries $\{U(\bm{x},\bm{\theta})= \prod_{\ell} V_\ell(\bm{\theta}) U_\ell(\bm{x})\}_{\bm{x}\in\mathbb{R}^d}$ (and particularly the unitaries $U_\ell(\bm{x})$ for $\ell \geq 2$). Note however that, in the case of our bit-string mappings, these unitaries are only implemented \emph{approximately} and that, in our gate-teleportation mappings, they are only implemented \emph{probabilistically}. Our gate-teleportation mappings are only exact from the point of view of \emph{models}, i.e., expectation values of observables, which are not covered by no-programming. The approximation errors $\delta>0$ and the normalization factors $(1-\delta')^{-1} \neq 1$ that we obtain in our mappings are indeed symptomatic of our inability to program data re-uploading unitaries both exactly and deterministically. On the other hand, our results show that, contrary to unitary maps, expectation values can be ``programmed'' exactly. 

\section{Explicit models are universal\\ function (family) approximators\label{sec:explicit-universal}}

From the universality of data re-uploading models as function approximators \cite{perez21,schuld21b} and our exact mappings from data re-uploading to explicit models, it trivially derives that explicit models are also universal function approximators. That is, for any integrable function $g\in\mathrm{L}([0,2\pi]^d)$ with a finite number of discontinuities, and for any $\varepsilon>0$, there exists an $n$-qubit feature encoding of the form $U_\phi(\bm{x}) = I \bigotimes_{i,j}R_z(2^{j}x_i)H$ and an observable $O$ such that $\abs{\text{Tr}[\rho(\bm{x})O] - g(\bm{x})} \leq \varepsilon, \forall \bm{x}$, for $\rho(\bm{x}) = U_{\phi}(\bm{x}) \ket{\bm{0}}\!\!\bra{\bm{0}} U_{\phi}^\dagger(\bm{x})$ (this result derives specifically from Theorem 2 in \cite{perez21}). A similar result was independently obtained in Ref.\ \cite{goto20}.

In this section, we show that this universal approximation property of explicit models also applies to computable hypothesis classes, i.e., function families of a known classical or quantum model. More precisely, we show that for any family $\{g_{\bm{\theta}}\}_{\bm{\theta}}$ of \emph{computable} functions $g_{\bm{\theta}}:[0,2\pi]^d \rightarrow \mathbb{R}$ specified by a Boolean or quantum circuit parametrized by a vector $\bm{\theta} \in [0,2\pi]^L$, and for any $\varepsilon > 0$, there exists an $n$-qubit feature encoding $U_\phi(\bm{x})$ using single-qubit rotations of the form $R_y(x_i)$ to encode $\bm{x}$, and a family of observables $O_{\bm{\theta}}$ parametrized by single-qubit rotations of the form $R_y(\theta_i)$, such that 
$\abs{\text{Tr}[\rho(\bm{x})O_{\bm{\theta}}] - g_{\bm{\theta}}(\bm{x})} \leq \varepsilon, \forall \bm{x}, \bm{\theta}$.

The proof of this result is rather simple. Using quantum amplitude (or phase) estimation \cite{brassard02} as a subroutine, we can, starting from a $\ket{\bm{0}}$ state and using $R_y(x_i)$ rotations, create a bit-string representation $\ket{\widetilde{\bm{x}}}$ of $\bm{x}$. This constitutes the feature encoding $U_{\phi}(\bm{x})$. On a different register, we create similarly a bit-string representation $\vert\widetilde{\bm{\theta}}\rangle$ of $\bm{\theta}$ using $R_y(\theta_i)$ rotations. These bit-strings, for an appropriately chosen precision of representation (which depends on the number of $R_y(x_i), R_y(\theta_i)$ rotations used), can then be used to approximate any computable function $g_{\bm{\theta}}(\bm{x})$ to some error $\varepsilon$.\footnote{Note that, for functions $g_{\bm{\theta}}$ that are computed using binary representations of $\bm{x}$ and $\bm{\theta}$, this construction can be made exact.} Indeed, when $g_{\bm{\theta}}$ is computed via a parametrized quantum circuit, we can use a similar construction to that depicted in Fig.~\ref{fig:construction} in the main text. When $g_{\bm{\theta}}$ is instead computed classically (e.g., using a neural network), we can either simulate this computation with a quantum circuit (see Sec.~3.2.5 of \cite{nielsen02}), or simply include it in the observables $O_{\bm{\theta}}$ as a post-processing of a computational basis measurement of $\ket{\widetilde{\bm{x}}}$ and $\vert\widetilde{\bm{\theta}}\rangle$.

The explicit models constructed in this proof may seem quite contrived and unnatural from the feature encodings and variational processing they use. Nonetheless, these constructions showcase how parametrized rotations, a natural building block to encode input data and to be used as variational gates, can leverage explicit quantum models to be universal function/model approximators.

\section{Beyond unitary feature encodings\label{sec:cptp-encoding}}

\begin{figure}[t]
\begin{center}
\includegraphics[width=0.55\linewidth]{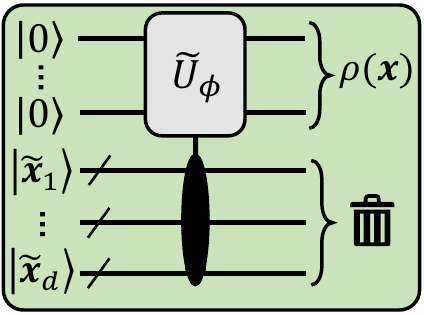}
\caption{A CPTP feature map based on bit-string encodings. Using data-independent controlled rotations, we can implement an approximation of any unitary feature encoding $U_\phi(\bm{x})$ by further tracing out the bit-string registers.}
\label{fig:construction-cptp}
\end{center}
\vspace{-2em}
\end{figure}

So far, in our definition of linear quantum models, we only considered \emph{unitary} feature encodings, i.e., where feature states are defined as $\rho(\bm{x}) = U_{\phi}(\bm{x}) \ket{\bm{0}}\!\!\bra{\bm{0}} U_{\phi}^\dagger(\bm{x})$ for a certain unitary map $U_{\phi}(\bm{x})$. In this section, we make the case that more general feature encodings, namely encodings for which the map $U_{\phi}(\bm{x})$ is allowed to be any completely positive trace preserving (CPTP) map, can lead to more interesting kernels $k(\bm{x},\bm{x'}) = \text{Tr}[\rho(\bm{x})\rho(\bm{x'})]$. This observation is in line with recent findings about quantum kernels derived from non-unitary feature encodings \cite{kubler21,huang20}.

We illustrate this point by focusing on the bit-string feature encoding $U_\phi(\bm{x})\ket{0^{\otimes n+dp}} = \ket{0^{\otimes n}}\bigotimes_{i=1}^{d} \ket{\widetilde{\bm{x}}_i}$ that we presented in the main text. We start by noting that augmenting this feature encoding with an arbitrary unitary $V$ always leads to the same kernel function:
\begin{align*}
    k(\bm{x},\bm{x'}) &= \abs{\bra{\bm{0}}U_\phi^\dagger(\bm{x'})V^\dagger VU_\phi(\bm{x})\ket{\bm{0}}}^2\\
    &= \abs{\bra{\bm{0}}U_\phi^\dagger(\bm{x'})U_\phi(\bm{x})\ket{\bm{0}}}^2
\end{align*}
given that $V^\dagger V = I$. If we however allow for a non-unitary operation such as tracing out part of the quantum system (which is allowed by CPTP maps), we can use this bit-string encoding to construct kernels $k(\bm{x},\bm{x'})$ that approximate \emph{virtually any} quantum kernel resulting from a unitary feature encoding on $n$ qubits. To see this, suppose for instance that we want to approximate the quantum kernel proposed by Havl\'{i}\v{c}ek \emph{et al.} \cite{havlivcek19} (resulting from the unitary feature encoding of Eq.~(\ref{eq:havlivcek})). In this case, we can, similarly to our construction in Fig.~\ref{fig:construction}, use data-independent rotations, controlled by the bit-string registers and acting on the $n$ working qubits, to simulate the data-dependent gates of Eq.~(\ref{eq:havlivcek}). Then, by tracing out the bit-string register, we effectively obtain an (arbitrarily good) approximation of the original feature encoding on the working qubits. This CPTP feature encoding is depicted in Fig.~\ref{fig:construction-cptp}.

\section{Details of the numerical simulations\label{sec:num-sim}}

In this section, we provide more details on the numerical simulations presented in the main text. We first describe how the training and testing data of the learning task are generated, then specify the quantum and classical models that we trained on this task.

\subsection{Dataset generation}

\subsubsection{Generating data points}

We generate our training and testing data by pre-processing the fashion MNIST dataset \cite{xiao17}. All $28\times28$-pixels grayscale images in the dataset are first subject to a dimensionality reduction via principal component analysis (PCA), where only their $n$ principal components are preserved, $2\leq n \leq 12$. This PCA gives rise to data vectors $\bm{x} \in \mathbb{R}^{n}$ that are further normalized component-wise as to each be centered around $0$ and have a standard deviation of $1$. To create a training set, we sample $M=1000$ of these vectors without replacement. A validation set and a test set, of size $M'=100$ each, are sampled similarly from the pre-processed fashion MNIST testing data.

\subsubsection{Generating labels}

The labels $g(\bm{x})$ of the data points $\bm{x}$ in the training, validation and test sets are generated using the explicit models depicted in Fig.~\ref{fig:vqc-fashion}, for a number of qubits $n$ equal to the dimension of $\bm{x}$, and uniformly random parameters $\bm{\theta}\in[0,2\pi]^{3nL}$.\\
The feature encoding takes the form \cite{havlivcek19}:
\begin{equation}\label{eq:havlivcek}
    U_\phi(\bm{x}) \ket{0^{\otimes n}} = U_z(\bm{x})H^{\otimes n}U_z(\bm{x})H^{\otimes n} \ket{0^{\otimes n}}
\end{equation}
for 
\begin{equation}
    U_z(\bm{x}) = \exp(-\mathrm{i}\pi\left[\sum_{i=1}^{n} x_i Z_i + \sum_{\substack{j=1,\\ j>i}}^{n} x_ix_j Z_iZ_j\right]).
\end{equation}
As for the variational unitaries $V(\bm{\theta})$, these are composed of $L$ layers of single-qubit rotations $R(\bm{\theta}_{i,j})$ on each of the qubits, interlaid with $CZ = \ket{1}\!\!\bra{1} \otimes Z$ gates between nearest neighbours in the circuit. We choose the number of layers $L$ as a function of the number of qubits $n$ in the circuit, such that the number of parameters ($3nL$) is approximately $90$ at all system sizes. Specifically, from $n = 2$ to $12$, we have $L = 15, 10, 7, 6, 5, 4, 4, 3, 3, 3, 3$, respectively.

Finally, the labels of the data points are specified by the expectations values
\begin{equation}\label{eq:generating-explicit}
    g(\bm{x}) = w_{\mathcal{D},\bm{\theta}}\text{Tr}[\rho(\bm{x}) V(\bm{\theta})^\dagger Z_1 V(\bm{\theta})]
\end{equation}
where $w_{\mathcal{D},\bm{\theta}}$ is a re-normalization factor that sets the standard deviation of these labels to $1$ over the training set.

\begin{figure}[t]
\includegraphics[width=\linewidth]{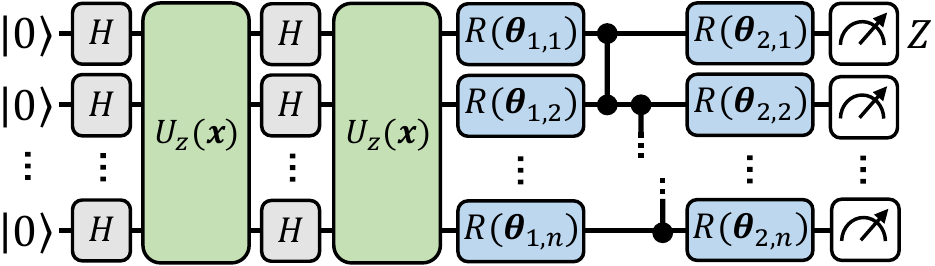}
\caption{The explicit model used in our numerical simulations. We use the feature encoding proposed by Havl\'{i}\v{c}ek \emph{et al.} \cite{havlivcek19} (see Eq.~(\ref{eq:havlivcek})), followed by a hardware-efficient variational circuit, where arbitrary single-qubit rotations $R(\bm{\theta}_{i,j}) = R_x(\theta_{i,j,0})R_y(\theta_{i,j,1})R_z(\theta_{i,j,2})$ on each qubit are interlaid with nearest-neighbour $CZ$ gates, for $L$ layers (here $L=2$). Finally, the expectation value of a $Z_1$ observable (with a re-normalization) assigns labels to input data $\bm{x} \in \mathbb{R}^n$.}
\label{fig:vqc-fashion}
\vspace{-1em}
\end{figure}

\subsubsection{Evaluating performance}
We evaluate the training loss of a hypothesis function $f$ using the mean squared error
\begin{equation}\label{eq:training-loss}
	\widehat{\mathcal{L}}(f) = \frac{1}{M}\sum_{m=1}^{M} \left( f(\bm{x}^{(m)})-g(\bm{x}^{(m)}) \right)^2
\end{equation}
on the training data. The test loss (indicative of the expected loss) is evaluated similarly on the test data (of size $M'$).

\subsection{Quantum machine learning models}

In our simulations, we compare the performance of two types of quantum machine learning models.

First, we consider explicit models from the same variational family as those used to label the data (i.e., depicted in Fig.~\ref{fig:vqc-fashion}), but initialized with \emph{different} variational parameters $\bm{\theta}\in[0,2\pi]^{3nL}$, now sampled according to a independent normal distributions $\mathcal{N}(0,0.05)$. As opposed to the generating functions of Eq.~(\ref{eq:generating-explicit}), we replace the observable weight $w_{\mathcal{D},\bm{\theta}}$ by a free parameter $w$, initialized to $1$ and trained along the variational parameters $\bm{\theta}$. We train the explicit models for $500$ steps of gradient descent on the training loss of Eq.~(\ref{eq:training-loss}). For this, we use an ADAM optimizer \cite{kingma14} with a learning rate $\alpha_{\bm{\theta}} = 0.01$ for the variational parameters $\bm{\theta}$ and a learning rate $\alpha_w = 0.1$ for the observable weight $w$.

Second, we also consider implicit models that rely on the same feature encoding $U_\phi(\bm{x})$ (Eq.~(\ref{eq:havlivcek})) as the explicit models. I.e., these take the form
\begin{equation}
    f_{\bm{\alpha},\mathcal{D}}(\bm{x}) = \text{Tr}[\rho(\bm{x}) O_{\bm{\alpha},\mathcal{D}}]
\end{equation}
for the same encodings $\rho(\bm{x}) = U_{\phi}(\bm{x}) \ket{\bm{0}}\!\!\bra{\bm{0}} U_{\phi}^\dagger(\bm{x})$, and an observable $O_{\bm{\alpha},\mathcal{D}}$ given by Eq.~(\ref{eq:implicit-obs}) in the main text. We train their parameters $\bm{\alpha}$ using the KernelRidge regression package of scikit-learn \cite{pedregosa11}. In the numerical simulations of Fig.~\ref{fig:fashion_regression}, we use an unregularized training loss, i.e., that of Eq.~(\ref{eq:training-loss}). The learning performance of the implicit models trained with regularization is presented in Appendix \ref{sec:add-num}.

\subsection{Classical machine learning models}

We additionally compare the performance of our quantum machine learning models to a list of classical models, identical to that of Huang \emph{et al.} \cite{huang20}. For completeness, we list these models here, and the hyperparameters they were trained with. All of these models were trained using the default specifications of scikit-learn \cite{pedregosa11}, unless stated otherwise.

\begin{itemize}[leftmargin=4mm]
    \item Neural network: We perform a grid search over two-layer feed-forward neural networks with hidden layers of size
    \begin{equation*}
        h \in \{ 10, 25, 50, 75, 100, 125, 150, 200 \}.
    \end{equation*}
    We use the MLPRegressor package with a maximum number of learning steps $\text{max\_iter} = 500$.
    \item Linear kernel method: We perform a grid search over the regularization parameter
    \begin{multline*}
        C \in \{0.006, 0.015, 0.03, 0.0625, 0.125, 0.25, 0.5, 1.0, 2.0,\\ 4.0, 8.0, 16.0, 32.0, 64.0, 128.0, 256, 512, 1024\}.
    \end{multline*}
    We select the best performance between the SVR and KernelRidge packages (both using the linear kernel).
    \item Gaussian kernel method: We perform a grid search over the regularization parameter
    \begin{multline*}
        C \in \{0.006, 0.015, 0.03, 0.0625, 0.125, 0.25, 0.5, 1.0, 2.0,\\ 4.0, 8.0, 16.0, 32.0, 64.0, 128.0, 256, 512, 1024\},
    \end{multline*}
    and the RBF kernel hyperparameter
    \begin{equation*}
        \gamma \in \{0.25, 0.5, 1.0, 2.0, 3.0, 4.0, 5.0, 20.0\}  / (n \text{Var}[\bm{x}]),
    \end{equation*}
    where $\text{Var}[\bm{x}]$ is the variance of all the components $x_i$, for all the data points $\bm{x}$ in the training set.\\
    We select the best performance between the SVR and KernelRidge packages (both using the RBF kernel).
    \item Random forest: We perform a grid search over the individual tree depth
    \begin{equation*}
        \text{max\_depth} \in \{2, 3, 4, 5\},
    \end{equation*}
    and the number of trees
    \begin{equation*}
        \text{n\_trees} \in \{25, 50, 100, 200, 500\}.
    \end{equation*}
    We use the RandomForestRegressor package.
    \item Gradient boosting: We perform a grid search over the individual tree depth
    \begin{equation*}
        \text{max\_depth} \in \{2, 3, 4, 5\},
    \end{equation*}
    and the number of trees
    \begin{equation*}
        \text{n\_trees} \in \{25, 50, 100, 200, 500\}.
    \end{equation*}
    We use the GradientBoostingRegressor package.
    \item Adaboost: We perform a grid search over the number of estimators
    \begin{equation*}
        \text{n\_estimators} \in \{25, 50, 100, 200, 500\}.
    \end{equation*}
    We use the AdaBoostRegressor package.
\end{itemize}

\begin{figure}[t]
\includegraphics[width=0.97\linewidth]{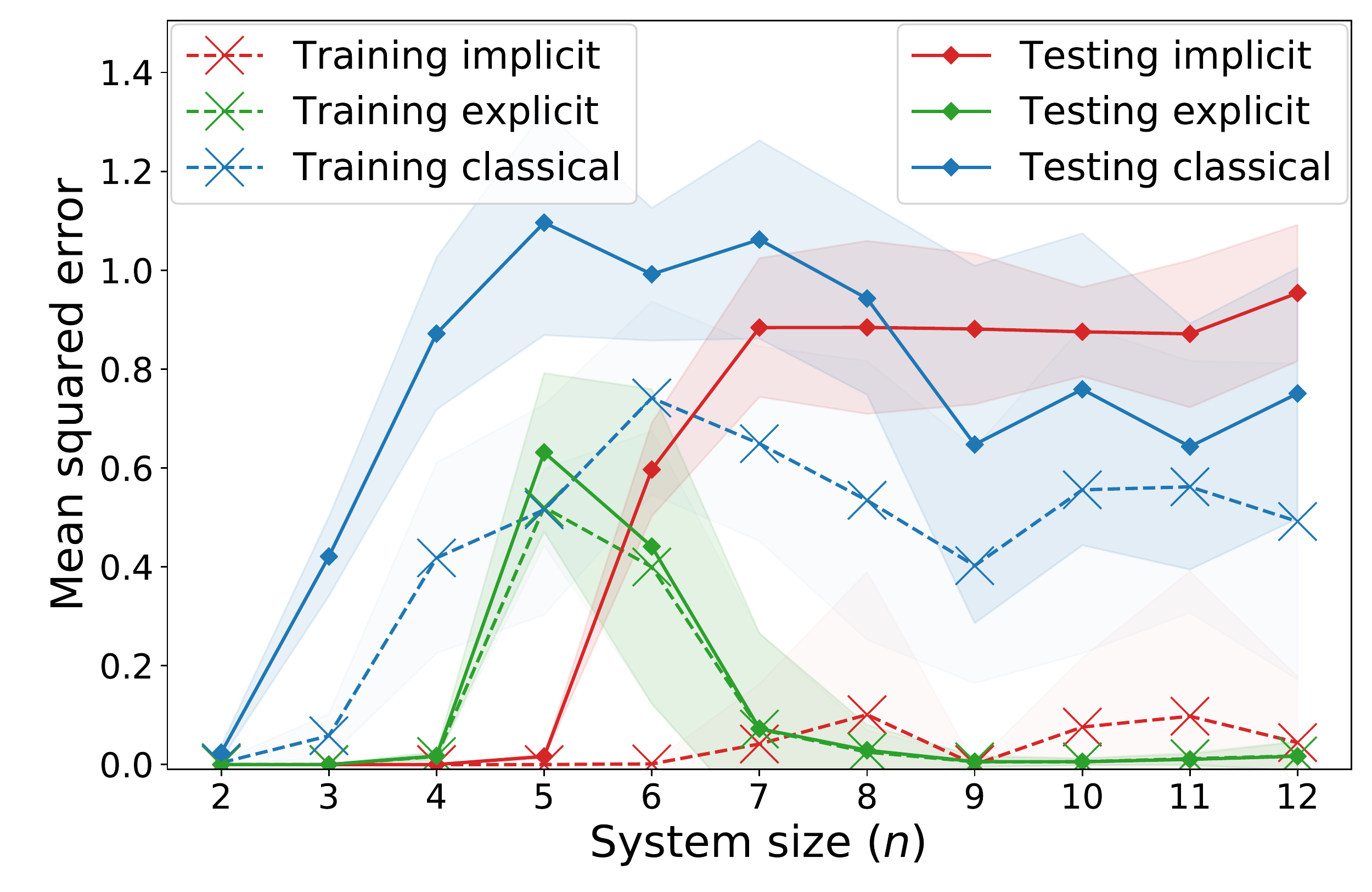}
\vspace{-1em}
\caption{Best performance of implicit models for different regularization strengths.}
\label{fig:fashion_regression2}
\vspace{-1em}
\end{figure}

At each system size, we keep the learning performance of the model with the lowest validation loss and plot its test loss.

\begin{figure}[t]
\includegraphics[width=0.97\linewidth]{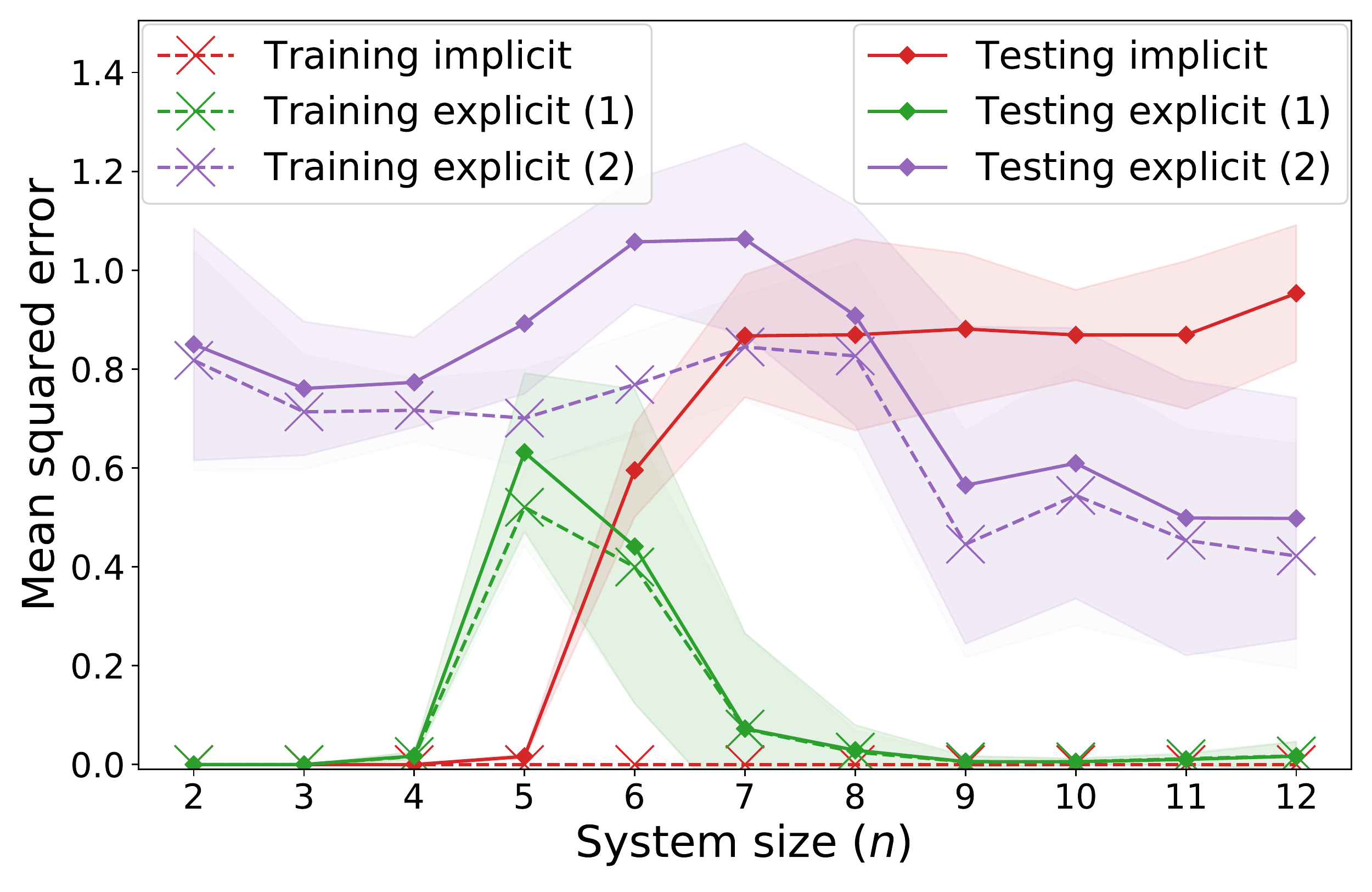}
\vspace{-1em}
\caption{Regression performance of explicit models from the same variational family as the models generating the data labels (1) and from a different variational family (2).}
\label{fig:fashion_regression3}
\vspace{-1em}
\end{figure}

\section{Additional numerical simulations\label{sec:add-num}}

In this section, we defer the results of our additional numerical simulations, in support of the claims made in the main text.

We first show that regularization does not improve the learning performance of the implicit models. In Fig.~\ref{fig:fashion_regression2} we plot the best test losses we obtained using regularization strengths $\lambda = 1/(2C)$ for $C \in \{0.006, 0.015, 0.03, 0.0625, 0.125, 0.25, 0.5, 1.0, 2.0,\allowbreak 4.0, 8.0, 16.0, 32.0, 64.0, 128.0, 256, 512, 1024\}$.

We also show that a variational family of observables not suited for the learning task can lead to poor learning performance. We illustrate this phenomenon by training explicit models constructed using the same feature encoding as those generating the data labels, but different variational unitaries $V(\bm{\theta})$ (see Eq.~(\ref{eq:generating-explicit})). We take these variational unitaries to resemble a Trotter evolution of a 1D-Heisenberg model with circular boundary conditions:
\begin{equation*}
    V(\bm{\theta}) = \prod_{i=1}^{L} \left( \prod_{j=1}^{n} e^{\mathrm{i}\theta_{i,j,0} Z_j Z_{j+1}} e^{\mathrm{i}\theta_{i,j,1} Y_j Y_{j+1}} e^{\mathrm{i}\theta_{i,j,2} X_j X_{j+1}} \right)
\end{equation*}
for the same number of layers $L$ and the same number of parameters $\bm{\theta} \in [0,2\pi]^{3nL}$ as the generating model. These are followed by a (weighted) $Z_1$ observable. Their learning performance is presented in Fig.~\ref{fig:fashion_regression3}.

\section{Learning separations between quantum models\label{sec:learn-sep}}

In this section, we establish rigorous learning separations between quantum learning models, using tools from recent works in classical machine learning \cite{daniely20,hsu21}. More specifically, we show the existence of a regression task specified by its input dimension $d \in \mathbb{N}$, such that, provably: (i) it can be solved exactly by a data re-uploading model acting on a single qubit and using $\mathcal{O}(\log(d))$ training samples, (ii) linear quantum models can also be sample efficient but require $\Omega(d)$ qubits to achieve a non-trivial expected loss, (iii) implicit models require both $\Omega(d)$ qubits and $\Omega(2^d)$ training samples to achieve this.

\subsection{Learning parity functions}

We consider the same learning task as Daniely \& Malach \cite{daniely20}, that is, learning $k$-sparse parity functions. These functions have discrete input and output spaces, $\mathcal{X} = \{-1,+1\}^d$ and $\mathcal{Y} = \{-1, +1\}$, respectively. $d \in \mathbb{N}$ specifies the dimension of the inputs $\bm{x} \in \mathcal{X}$, and an additional parameter $0 \leq k \leq d$ specifies the family of so-called $k$-sparse parity functions: $$g_A(\bm{x}) = \prod_{i \in A}x_i$$ for $A \subset [d]$ and $|A| = k$. That is, for a given subset $A$ of the input components $[d]$, the function $g_A$ returns the parity $\pm 1$ of these components for any input $\bm{x} \in \mathcal{X}$. $A$ can be any subset of $[d]$ of size $k$, which gives us a family of functions $\{ g_A \}_{A \subset [d], \abs{A} = k}$ that we take to be our concept class (for a certain $k$ specified later). These functions have the interesting property that they are all linearly independent (despite potentially being an exponentially large family), which is the essential property we'll be using to derive our separation results.

Daniely \& Malach \cite{daniely20} show that, for an appropriate choice of input distribution and loss function, $k$-sparse parity functions cannot be approximated by any \emph{polynomial-size} linear model (for $k \leq \frac{d}{16}$), while a depth-2 neural network with hidden layers polynomially large in $k$ can learn these (almost) perfectly. The size of the linear model is defined as the dimension of its feature space $\mathcal{F}$ multiplied by the norm of its weight vector $\norm{w}_{\mathcal{F}}$.

In our results, we rely instead on a powerful theorem by Hsu \cite{hsu21} to derive simpler separations results in terms of the \emph{dimension} of linear models, rather than their size. This theorem, stated in the next subsection, makes it natural to consider an input distribution $\mathcal{D_{\mathcal{X}}}$ that is simply the uniform distribution over the data space $\mathcal{X}$, and the mean-squared error $\mathcal{L}_{A}(f) = \mathbb{E}_{\bm{x}\sim\mathcal{D}_{\mathcal{X}}}(f(\bm{x})-g_{A}(\bm{x}))^2$ to be the expected loss for which we establish our bounds. 

\subsection{Learning performance of quantum models}

This section is organized as follows: we first describe our lower bound results for linear quantum models, which are naturally derived for learning parities w.r.t.~to the uniform input distribution $\mathcal{D}_{\mathcal{X}}$. We then move to our upper bound results for data re-uploading models, which make us consider a different input distribution $\mathcal{D}_{A}$ in order to achieve a largest separation with linear models. Together, these bounds give us our main Theorem \ref{thm:sep-mother}.

\begin{figure*}[t]
\vspace{-1em}
\includegraphics[width=0.97\linewidth]{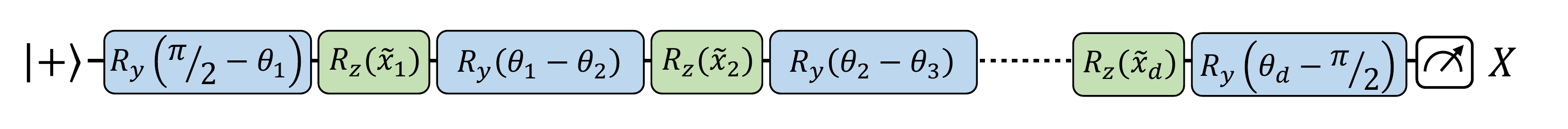}
\vspace{-1em}
\caption{Data re-uploading circuit used to learn parity functions. The encoding angles are $\widetilde{x}_i = \frac{\pi}{2}(x_i-1)$, $\forall i = 1, \ldots, d$. Note that $R_y(\frac{\pi}{2})$ implements the transformation $\ket{\pm} \rightarrow \ket{1/0}$ and $R_y(-\frac{\pi}{2})$ implements $\ket{1/0} \rightarrow \ket{\pm}$, such that, to represent any parity function, we only need $\theta_i \in \{0,\frac{\pi}{2}\}, \forall i = 1, \ldots, d$.}
\label{fig:parity-circuit}
\vspace{-1em}
\end{figure*}

\subsubsection{Lower bounds for linear models}

As mentioned above, our separation results derive from dimensionality arguments on the family of functions that can be represented by linear models. To start, let us consider the Hilbert space $\mathcal{H}= L^2(\mathcal{D}_{\mathcal{X}})$ of real-valued functions that are square-summable with respect to the probability space $(\mathcal{X},2^{\mathcal{X}}, \mathcal{D}_{\mathcal{X}})$. The inner product associated to this Hilbert space is $\langle f,g \rangle_{\mathcal{H}}=\frac{1}{|\mathcal{X}|} \sum_{\bm{x} \in \mathcal{X}} f(\bm{x})g(\bm{x})$. For any $k \in \{0, \ldots, d\}$, $k$-sparse parity functions belong to $\mathcal{H}$, and they are moreover orthogonal functions of this Hilbert space. As for the functions generated by a linear model $f(\bm{x}) = \langle\phi(\bm{x}), w \rangle_{\mathcal{F}}$ (for a fixed feature encoding $\phi$), these are also functions of the Hilbert space $\mathcal{H}$. More specifically, they are contained in a finite-dimensional subspace $\mathcal{W} = \text{span}\{ \langle\phi(\bm{x}), w_j \rangle_{\mathcal{F}} \}_{w_j}$ of $\mathcal{H}$, for $\{w_j\}_j$ a basis of $\mathcal{F}$ (with respect to the inner product $\langle \cdot, \cdot \rangle_{\mathcal{F}}$).\\
To establish our bounds, all we need to do now is: a) relate the dimension $\dim(\mathcal{W})$ of this subspace to the expected loss of the linear model, and b) upper bound $\dim(\mathcal{W})$ in terms of the number of qubits or data samples accessible to this model.

To get a), we use directly the following theorem:
\begin{theorem}[Theorem 1 in \cite{hsu21}, Theorem 29 in \cite{hsu21b}]\label{thm:dim-hsu}
Let $\varphi_1, \ldots, \varphi_N$ be orthogonal functions in a Hilbert space $\mathcal{H}$ such that $\norm{\varphi_i}^2_{\mathcal{H}} = 1$ for all $i = 1, \ldots, N$, and let $\mathcal{W}$ be a finite subspace of $\mathcal{H}$. Then, for:
$$\varepsilon = \frac{1}{N}\sum_{i=1}^N \left[ \inf_{f \in \mathcal{W}} \norm{f-\varphi_i}^2_{\mathcal{H}} \right], $$
we have
$$\textnormal{dim}(\mathcal{W}) \geq N(1-\varepsilon).$$
\end{theorem}
By definition here, $\norm{f-\varphi_i}^2_{\mathcal{H}} = \frac{1}{|\mathcal{X}|} \sum_{\bm{x} \in \mathcal{X}} (f(\bm{x})-\varphi_i(\bm{x}))^2$ is the mean-squared error $\mathcal{L}_i (f)$ of $f$ with respect to $\varphi_i$.
Given that, for $\{\varphi_i\}_{i\in[N]} = \{ g_A \}_{A \subset [d], \abs{A} = k}$, we have $N = \binom{d}{k}$, this theorem gives us a combinatorial lower bound on $\textnormal{dim}(\mathcal{W})$ when the linear model achieves a non-trivial average expected loss $\varepsilon < 1$ (otherwise obtained for $f(\bm{x})=0,\ \forall \bm{x} \in \mathcal{X}$).\\

We now move to point b). Note here that, in order to upper bound $\textnormal{dim}(\mathcal{W})$, all we need to do is upper bound the number of independent vectors in the $\text{span}\{ \langle\phi(\bm{x}), w_j \rangle_{\mathcal{F}} \}_{w_j} = \mathcal{W}$, hence, equivalently, the number of basis vectors $w_j$ of $\mathcal{F}$. This can easily be done when we know the number of qubits on which the linear model acts. Indeed, for an $n$-qubit model, $\mathcal{F}$ is the space of $2^n \times 2^n$ Hermitian operators, and hence $\textnormal{dim}(\mathcal{W}) \leq 2^{2n}$. This leads us to our first lemma.

\begin{lemma}\label{lem:sep-linear}
There exists a regression task specified by an input dimension $d \in \mathbb{N}$, a function family $\{g_A : \{-1,1\}^d \rightarrow \{-1,1\}\}_{A\subset[d], |A|=\lfloor d/2 \rfloor}$, and an input distribution $\mathcal{D}_{\mathcal{X}}$, such that, to achieve an average mean-squared error
$$\mathbb{E}_{A} \big[ \inf_{f \in \mathcal{W}} \norm{f-g_A}^2_{L^2(\mathcal{D}_{\mathcal{X}})} \big] \leq \varepsilon$$
any linear quantum model needs to act on
$$n \geq \frac{d}{4} + \frac{1}{2}\log_2(1-\varepsilon)$$
qubits.
\end{lemma}
\begin{proof}
For $k = |A| = \lfloor\frac{d}{2}\rfloor$, we have $N = \binom{d}{k} \geq 2^{d/2}$. From Theorem \ref{thm:dim-hsu}, we have $\textnormal{dim}(\mathcal{W}) \geq 2^{d/2}(1-\varepsilon)$, and from our previous observation, $\textnormal{dim}(\mathcal{W}) \leq 2^{2n}$.
\end{proof}

In the case of implicit linear models, note that we can bound $\textnormal{dim}(\mathcal{W})$ even more tightly. This is because the weight vector $w \in \mathcal{F}$, or equivalently the observable of the implicit quantum model, is expressed as a linear combination of embedded data samples $\phi(\bm{x}^{(i)})=\rho(\bm{x}^{(i)})$. Therefore, the number of independent vectors in the $\text{span}\{ \langle\phi(\bm{x}), w_j \rangle_{\mathcal{F}} \}_{w_j} = \mathcal{W}$ is upper bounded by the number of data points $M$ in the training set of the implicit model. This gives us $\textnormal{dim}(\mathcal{W}) \leq \min(2^{2n}, M)$ and the following lemma:
\begin{lemma}\label{lem:sep-implicit}
There exists a regression task specified by an input dimension $d \in \mathbb{N}$, a function family $\{g_A : \{-1,1\}^d \rightarrow \{-1,1\}\}_{A\subset[d], |A|=\lfloor d/2 \rfloor}$, and an input distribution $\mathcal{D}_{\mathcal{X}}$, such that, to achieve an average mean-squared error
$$\mathbb{E}_{A} \big[ \inf_{f \in \mathcal{W}} \norm{f-g_A}^2_{L^2(\mathcal{D}_{\mathcal{X}})} \big] \leq \varepsilon$$
any implicit quantum model needs to use
$$M \geq 2^{d/2}(1-\varepsilon)$$
data samples.
\end{lemma}

Note that implicit quantum models suffer from both lower bounds in Lemmas \ref{lem:sep-linear} and \ref{lem:sep-implicit}: they require both $\Omega(d)$ qubits \emph{and} $\Omega(2^d)$ data samples.

\subsubsection{Upper bound for data re-uploading models}

To establish our learning separation, we would like to show that data re-uploading circuits can efficiently represent parity functions. We show this constructively, by designing a single-qubit data re-uploading model that can compute any ($k$-sparse) parity function. This model is depicted in Fig.~\ref{fig:parity-circuit} and consists solely in $R_z(\widetilde{x}_i) = Z^{(x_i-1)/2}$ encoding gates, parametrized $R_y$ rotations and a final Pauli-X measurement. To understand how such a circuit can compute parity functions, consider the parity to be encoded in the qubit being either in a $\ket{+}$ or a $\ket{-}$ state. Therefore, the $R_z(\widetilde{x}_i)$ rotations flip the $\ket{\pm}$ state whenever $x_i = -1$, and preserve it otherwise. As for the $R_y(\pm\frac{\pi}{2})$ rotations, these essentially act as Hadamard gates by transforming a $\ket{\pm}$ state into a $\ket{1/0}$ state and back, which allows to ``hide'' this state from the action of a $R_z(\widetilde{x}_i)$ gate. We parametrize these gates as to hide the $\ket{\pm}$ qubit from $R_z(\widetilde{x}_i)$ whenever $\theta_i = 0$, and let it act whenever $\theta_i = \frac{\pi}{2}$. This leads us to the following lemma:

\begin{lemma}\label{lem:sep-reuploading}
For the same learning task considered in Lemmas \ref{lem:sep-linear} and \ref{lem:sep-implicit}, there exists a data re-uploading model acting on a single qubit, and with depth $2d+1$, that achieves a perfect mean-squared error:
$$\mathbb{E}_{A} \big[ \min_{f} \norm{f-g_A}^2_{L^2(\mathcal{D}_{\mathcal{X}})} \big] = 0.$$
\end{lemma}
\begin{proof}
For a given $A \subset [d]$, take in the circuit of Fig.~\ref{fig:parity-circuit}: $\theta_{i} = \pi/2$ if $i \in A$ and $\theta_i = 0$ otherwise. 
\end{proof}

Note that our claim on data re-uploading models so far only has to do with representability and not actual learning from data. We are yet to prove that a similar learning performance can be achieved from a training set of size polynomial in $d$ and a polynomial-time learning algorithm. For the uniform data distribution $\mathcal{D}_{\mathcal{X}}$ we considered so far, this is known to be possible using $\mathcal{O}(d)$ data samples and by solving a resulting linear system of equations acting on $d$ variables \cite{kearns98}. However, this distribution does not provide us with the best possible separation in terms of data samples, which is why we consider instead the \emph{mixture} of data distributions introduced by Daniely \& Malach \cite{daniely20}. Originally intended to get around the hardness of learning parities with gradient-based algorithms \cite{shalev17}, this distribution significantly reduces the data requirements for the data re-uploading (and explicit linear) models to $\mathcal{O}(\log(d))$, while preserving the $\Omega(2^d)$ lower bound for implicit models.

For every $A \subset [d]$, call $\mathcal{D}^{(1)}_A = \mathcal{D}_{\mathcal{X}}$ the uniform distribution over $\mathcal{X}$, and $\mathcal{D}^{(2)}_A$ the distribution where all components in $[d] \setminus A$ are drawn uniformly at random, while, independently, the components in $A$ are all $+1$ with probability $1/2$  and all $-1$ otherwise. The data distribution $\mathcal{D}_A$ that we consider samples $\bm{x}\sim\mathcal{D}^{(1)}_A$ with probability $1/2$ and $\bm{x}\sim\mathcal{D}^{(2)}_A$ with probability $1/2$. For $k = \abs{A}$ an odd number\footnote{In Lemma \ref{lem:sep-linear}, when $\lfloor\frac{d}{2}\rfloor$ is an even number, we can take $k = \lfloor\frac{d}{2}\rfloor +1$, for which $\binom{d}{k} \geq 2^{d/2}$ still holds.}, this distribution is particularly interesting as, when $\bm{x}\sim\mathcal{D}^{(2)}_A$, $x_i = g_A(\bm{x})$ for all $i \in A$, which statistically ``reveals'' $A$, while $\mathcal{D}^{(1)}_A$ still preserves our previous hardness of generalization results. This allows us to prove the following lemma.

\begin{lemma}\label{lem:training-reuploading}
For the data distribution $\mathcal{D}_A$ defined above, there exists a learning algorithm using 
$$M = 32 \log(\frac{2d}{\delta})$$
data samples and $dM$ evaluations of the circuit in Fig.~\ref{fig:parity-circuit}, that returns, for any $A \subset [d]$, with odd $\abs{A} = k$, a function $f_A$ satisfying 
$$\norm{f_A-g_A}^2_{L^2(\mathcal{D}_{\mathcal{X}})}= 0$$
with probability at least $1-\delta$.
\end{lemma}
\begin{proof}
We analyze the following learning algorithm: given a training set of size $M$, evaluate, for all $i \in [d]$, the empirical loss $\widehat{\mathcal{L}}(f_i)$ of the data re-uploading function $f_i$ obtained with the parameters $\theta_i = \frac{\pi}{2}$, $\theta_j = 0$ for $j \neq i$. Return $\theta_i = \frac{\pi}{2}$ when $\widehat{\mathcal{L}}(f_i) \leq 1.5$ and $\theta_i = 0$ otherwise, for all $i \in [d]$.

Call $X_i = (x_i - g_A(\bm{x}))^2$ the random variable obtained by sampling $\bm{x}$ from the data distribution $\mathcal{D}_A$, for a given $A$. Note that, by construction, $\mathbb{E}_{\mathcal{D}^{(1)}_A}[X_i] = 2$ and $\mathbb{E}_{\mathcal{D}^{(2)}_A}[X_i] = 0$ for $i \in A$, while $\mathbb{E}_{\mathcal{D}^{(1)}_A}[X_i] = \mathbb{E}_{\mathcal{D}^{(2)}_A}[X_i]= 2$ for $i \notin A$. Therefore $\mathbb{E}_{\mathcal{D}_A}[X_i] = 1$ for $i \in A$, and $\mathbb{E}_{\mathcal{D}_A}[X_i] = 2$ otherwise.\\
Given that the computed losses $\widehat{\mathcal{L}}(f_i)$ are empirical estimates of $\mathbb{E}_{\mathcal{D}_A}[X_i]$, all we need in order to identify $A$ is guarantee with high probability that we can distinguish whether $\mathbb{E}_{\mathcal{D}_A}[X_i] = 1$ or $2$, for all $i\in[d]$. We achieve this guarantee using the union bound and Hoeffding's inequality ($X_i \in [0,4]$):
\begin{align*}
P\Bigg(\bigcup_{i=1}^d \bigg(|\widehat{\mathcal{L}}(f_i)&-\mathbb{E}[X_i]|\geq\frac{1}{2}\bigg)\Bigg)&& \\
& \leq d P\left(|\widehat{\mathcal{L}}(f_i)-\mathbb{E}[X_i]|\geq\frac{1}{2}\right)\\
& \leq 2d \exp(-\frac{M}{32})
\end{align*}
and to upper bound this failure probability by $\delta$, we need $M \geq 32 \log(\frac{2d}{\delta})$.
\end{proof}

We leave as an open question whether an optimization procedure with similar learning guarantees but based on gradient descent also exists.

\subsubsection{Main theorem}

To conclude our results, we are left to show that similar lower bounds to those in Lemmas \ref{lem:sep-linear} and \ref{lem:sep-implicit} also hold for the data distribution $\mathcal{D}_{A}$. Intuitively, this problem is easy to solve: given the inability of linear models to \emph{represent} good approximations of parity functions with respect to the uniform data distribution, it should be clear that these still have a bad generalization performance with respect to $\mathcal{D}_{A}^{(1)}$, despite the information revealed by $\mathcal{D}_{A}^{(2)}$. We make this intuition rigorous in the following theorem (restatement of Theorem \ref{thm:separation-main} in the main text).

\begin{theorem}\label{thm:sep-mother}
There exists a regression task specified by an input dimension $d \in \mathbb{N}$, a function family $\{g_A : \{-1,1\}^d \rightarrow \{-1,1\}\}_{A\subset[d], |A|=\lfloor d/2 \rfloor}$, and associated input distributions $\mathcal{D}_{A}$, such that, to achieve an average mean-squared error
$$\mathbb{E}_{A} \big[ \inf_{f \in \mathcal{W}} \norm{f-g_A}^2_{L^2(\mathcal{D}_{A})} \big] = \varepsilon < \frac{1}{2}$$
(i) any linear quantum model needs to act on
$$n \geq \frac{d}{4} + \frac{1}{2}\log_2(1-2\varepsilon)$$
qubits,\\
(ii) any implicit quantum model additionally requires
$$M \geq 2^{d/2}(1-2\varepsilon)$$
data samples, while\\
(iii) a data re-uploading model acting on a single qubit can be trained to achieve a perfect expected loss with probability $1-\delta$, using $M = 32 \log(\frac{2d}{\delta})$ data samples.
\end{theorem}
\begin{proof}
We relate $\varepsilon$ to $\varepsilon_{{A}^{(1)}}$ (defined similarly to $\varepsilon$, but with respect to $\mathcal{D}_{A}^{(1)}$):
\begin{align*}
\varepsilon &= \mathbb{E}_{A} [ \inf_{f \in \mathcal{W}} \norm{f-g_A}^2_{L^2(\mathcal{D}_{A})} ]\\
&= \mathbb{E}_{A} [ \inf_{f \in \mathcal{W}} \sum_{\bm{x}\in\mathcal{X}} \mathcal{D}_{A}(\bm{x}) (f(\bm{x}) - g_A(\bm{x}))^2 ]\\
&\geq \frac{1}{2} \mathbb{E}_{A} [ \inf_{f \in \mathcal{W}} \sum_{\bm{x}\in\mathcal{X}} \mathcal{D}_{A}^{(1)}(\bm{x}) (f(\bm{x}) - g_A(\bm{x}))^2 ]\\
&\quad + \frac{1}{2} \mathbb{E}_{A} [ \inf_{f \in \mathcal{W}} \sum_{\bm{x}\in\mathcal{X}} \mathcal{D}_{A}^{(2)}(\bm{x}) (f(\bm{x}) - g_A(\bm{x}))^2 ]\\
&\geq \frac{1}{2} \varepsilon_{{A}^{(1)}}
\end{align*}
since $\varepsilon_{{A}^{(2)}} \geq 0$.\\
From Lemma \ref{lem:sep-linear}, we have that, to ensure $\varepsilon_{A^{(1)}} \leq 2 \varepsilon$, we need at least $\frac{d}{4} + \frac{1}{2}\log_2(1-2\varepsilon)$ qubits, which proves (i). (ii) follows similarly from Lemma \ref{lem:sep-implicit}.\\
Point (iii) corresponds to Lemmas \ref{lem:sep-reuploading} and \ref{lem:training-reuploading}.
\end{proof}

\subsection{Tight bounds on linear realizations of data re-uploading models}
A direct corollary of Lemmas \ref{lem:sep-linear} and \ref{lem:sep-reuploading} is a lower bound on the number of additional qubits required to map any data re-uploading model to an equivalent (explicit) linear model. Indeed, since the data re-uploading model of Fig.~\ref{fig:parity-circuit} can represent any parity function exactly for any input dimension $d\in\mathbb{N}$, while a linear model using a number of qubits sublinear in $d$ can only achieve poor approximations on average, we can easily prove the following theorem (Corollary \ref{cor:lower-bound} in the main text).

\begin{theorem}
Any procedure that takes as input an arbitrary data re-uploading model $f_{\bm{\theta}}$ with $d$ encoding gates and returns an equivalent explicit model $\widetilde{f}_{\bm{\theta}}$ (i.e., a universal mapping) must produce models acting on $\Omega(d)$ additional qubits for worst-case inputs.
\end{theorem}
\begin{proof}
By contradiction: were there a universal mapping using only $\widetilde{\mathcal{O}}(d^{1-\alpha})$, $\alpha>0$ additional qubits, (i.e, sublinear in $d$), applying it to the circuit in Fig.~\ref{fig:parity-circuit} would result in a linear model acting on $\widetilde{\mathcal{O}}(d^{1-\alpha})$ qubits with perfect performance in representing parity functions, which contradicts Lemma \ref{lem:sep-linear}.
\end{proof}

Note that our gate-teleportation mapping has an overhead of $\mathcal{O}(d\log(d/\delta'))$ qubits where $\delta'$ is a controllable parameter. This theorem hence proves that our mapping is \emph{essentially} optimal with respect to this overhead.

\subsection{The case of classification}

So far, in our separation results, we have only considered a \emph{regression} loss (the mean-squared error), despite the parity functions having a discrete output. It is an intriguing question whether similar separation results can be obtained in the case of \emph{classification}, i.e., for a binary classification loss $\mathbb{E}_{\bm{x}\sim\mathcal{D}_{\mathcal{X}}}|\text{sign}(f(\bm{x}))-g_{A}(\bm{x})|$ for instance.

It is rather straightforward to show a $\Omega(d)$ qubits lower bound for linear classifiers that achieve \emph{exact} learning (i.e., a $0$ loss). We can consider here the concept class of all $k$-sparse parity functions for $k \in \{0, \ldots, d\}$, such that it contains all possible labelings $\mathcal{X} \rightarrow \{-1,1\}$. Therefore, a model that can represent all these functions exactly needs, by definition, a VC dimension larger than $2^d$. But we know that, for linear quantum classifiers acting on $n$ qubits, this VC dimension is upper bounded by $2^{2n}+1$ \cite{gyurik21}. The lower bound then trivially derives.

Making this lower bound \emph{robust} (i.e., allowing an $\varepsilon\geq0$ loss) is however a harder task. By noting that a $\Omega(d)$ lower bound on the feature space dimension of a linear classifier implies a $\Omega(\log(d))$ lower bound on the number of qubits of a quantum linear classifier ($\textnormal{dim}(\mathcal{F}) = 2^{2n}$), we can adapt a result from Kamath \emph{et al.} \cite{kamath20} to show the following: to achieve an average classification error $\varepsilon$ on $1$-sparse parity functions (i.e., $f_i(\bm{x})= x_i$), a linear quantum classifier needs to act on $\Omega(\log[d(1-H(\varepsilon))])$ qubits (for $H(\varepsilon)$ the binary entropy function). However, according to the same authors, establishing a stronger lower bound (e.g., number of qubits poly-logarithmic in $d$, or equivalently, a feature space dimension super-polynomial in $d$) for a similar task would constitute a major frontier in complexity theory. Such a result would provide with a function that requires a depth-2 threshold circuit of super-polynomial size to be computed, while the best known lower bounds for the size of depth-2 threshold circuit are only polynomial.

\end{document}